\documentclass[10pt,journal,compsoc]{IEEEtran}
\usepackage{url}
\usepackage{verbatim}
\usepackage{cite}
\usepackage{amsmath,amssymb,amsfonts}
\usepackage{lipsum}
\usepackage{caption}
\usepackage{graphicx}
\usepackage{textcomp}
\usepackage{xcolor}
\usepackage{subfigure}
\usepackage{amsthm}
\usepackage{algorithmic}
\usepackage{algorithm}

\newtheorem{theorem}{Theorem}  [section]

\ifCLASSINFOpdf

\else

\fi

\begin{document}
\title{Resource Optimization for Blockchain-based Federated Learning in Mobile Edge Computing
\author{Zhilin Wang, Qin Hu, Zehui Xiong}
\thanks{This work is partly supported by the US NSF under grant CNS-2105004.}
\IEEEcompsocitemizethanks{
\IEEEcompsocthanksitem Zhilin Wang and Qin Hu are with the Department of Computer \& Information Science, 
Indiana University-Purdue University Indianapolis,  USA. The corresponding author is Qin Hu. Email: \{wangzhil, qinhu\}@iu.edu\\
\IEEEcompsocthanksitem Zehui Xiong is with the Pillar of Information Systems Technology \& Design, Singapore University of Technology Design, Singapore. Email: zehui\_xiong@sutd.edu.sg
}
}


\IEEEtitleabstractindextext{%
\begin{abstract}
With the development of mobile edge computing (MEC) and blockchain-based federated learning (BCFL), a number of studies suggest deploying BCFL on edge servers. In this case, resource-limited edge servers need to serve both mobile devices for their offloading tasks and the BCFL system for model training and blockchain consensus in a cost-efficient manner without sacrificing the service quality to any side. To address this challenge, this paper proposes a resource allocation scheme for edge servers, aiming to provide the optimal services with the minimum cost. Specifically, we first analyze the energy consumed by the MEC and BCFL tasks, and then use the completion time of each task as the service quality constraint. Then, we model the resource allocation challenge into a multivariate, multi-constraint, and convex optimization problem. To solve the problem in a progressive manner, we design two algorithms based on the alternating direction method of multipliers (ADMM) in both the homogeneous and heterogeneous situations with equal and on-demand resource distribution strategies, respectively. The validity of our proposed algorithms is proved via rigorous theoretical analysis.
Through extensive experiments, the convergence and efficiency of our proposed resource allocation schemes are evaluated. To the best of our knowledge, this is the first work to investigate the resource allocation dilemma of edge servers for BCFL in MEC.
\end{abstract}

}

\maketitle

\IEEEdisplaynontitleabstractindextext

\IEEEpeerreviewmaketitle

\section{Introduction}
\IEEEPARstart{A}S embedded sensors are widely deployed on mobile devices, such as smartphones and smart vehicles, they can pervasively perceive the physical world and collect an extensive amount of data. With the advances of hardware technology, it becomes promising for devices to process the collected data locally, such as training machine learning models. However, as the resources of mobile devices are usually inadequate, they may experience difficulty finishing computing-intensive tasks, 
which drives the emergence of mobile edge computing (MEC). 
Its basic idea is to facilitate mobile devices offloading computing tasks to their nearby edge servers with sufficient resources and then obtain the calculated results with communication-efficiency in close proximity \cite{abbas2017mobile, liang2017mobile}. 
MEC has been applied to many fields, such as the Internet of Things (IoT) \cite{sun2016edgeiot, sabella2016mobile}, smart healthcare \cite{sodhro2019mobile, li2019edgecare}, and smart transportation \cite{chen2019deep, zhou2019lightweight}. 

To address the main challenges of federated learning (FL)\cite{mcmahan2017communication, bonawitz2019towards}, such as the single point of failure and the privacy protection of model updates, 
blockchain has been extensively used to assist in achieving full decentralization with security \cite{zhao2020privacy,pokhrel2020federated}, which is termed as Blockchain-based FL (BCFL).
This new framework connects participants in FL, i.e., clients, through blockchain network and requires them to complete both FL and blockchain related operations, such as data collection, model training, and block generation \cite{wang2021blockchain}. As for a client in BCFL, a large amount of resources will be consumed in completing the BCFL task, however, making it an impractical job for mobile devices with constrained resources. To address this issue, researchers advocate deploying BCFL at the edge given that edge servers usually have strong computing, communication, and storage capabilities for FL model training and blockchain consensus \cite{zhao2020mobile,hu2021blockchain}.


In this case, the MEC servers are responsible for completing both the BCFL and MEC tasks. 
For the MEC task, the edge server is usually required to allocate the communication resource (e.g., bandwidth) for mobile devices to transfer data, the storage resource for data caching, and the computing resources (e.g., CPU cycle frequency) for computing based on the received data. Similarly, for the BCFL task, the edge server needs to distribute the bandwidth resource for sharing model updates  and reaching consensus among blockchain nodes, the storage resource for saving the copy of blockchain data and local training data, and the CPU resource for FL model training and updating, as well as the generation of new blocks. 
Since both the MEC and BCFL tasks are time-sensitive and could be performed simultaneously at the MEC servers, the servers have to deal with 
the challenge of serving both the lower-layer mobile edge devices and the upper-layer BCFL system without significant delay, 
which makes it a necessity to design reasonable resource allocation schemes for them.


The existing resource allocation mechanisms for BCFL and MEC can be classified into two main categories, where the first type of studies focus on allocating resources to each device for the requested MEC task \cite{sardellitti2015joint,wang2019smart,wan2019task}, and the other is to allocate resources for model training and block generation in the BCFL task \cite{hieu2020resource,li2021blockchain}.
Although the state-of-the-art studies can help edge servers allocate resources to handling the MEC or BCFL tasks, these mechanisms have never considered the resource conflicts when both tasks are running on servers at the same time. 

To fill the gap, we design a resource allocation scheme that allows the edge server to finish both the MEC and BCFL tasks simultaneously and timely.
Specifically, we define the cost as the total energy consumed by the edge server in completing both the MEC and BCFL tasks, and then use the corresponding time requirements as the constraints on the quality of services provided by the edge server.
We can transform the resource allocation problem into a multivariate, multi-constraint, and convex optimization problem. However, solving this optimization problem faces the following challenges: 
1) there are multiple variables since assigning resources to the MEC task means making decisions on resource allocation to each device, where the number of variables increases  with the device quantity; 
and 2) there are too many constraints, making the solution non-trivial.



These two challenges result in the invalidity of traditional optimization methods for multiple variable calculation. Therefore, we design a scheme based on a distributed optimization algorithm, named the alternating direction method of multipliers (ADMM), which is the combination of \textit{dual ascent} and \textit{dual decomposition}, and can determine multiple variables by iterations in a distributed manner \cite{boyd2011distributed}. We adopt the ADMM-based algorithm to solve our proposed resource allocation problem in two scenarios progressively for easy understanding. 
Specifically, we first apply the \textit{modified general ADMM} (MG-ADMM) algorithm for the homogeneous scenario, which distributes resources to each local device equally; then we propose the \textit{modified consensus ADMM} (MC-ADMM) based algorithm to assign resources to devices on demand in the heterogeneous scenario. Besides, by adding additional regularization terms during the variable iteration process, our proposed algorithms can converge in the case of more than two variables, which is confirmed by theoretical analysis. Finally, we conduct extensive experiments to testify the convergence and the effectiveness of our proposed resource allocation schemes.



To the best of our knowledge, we are the  first  to tackle  the  challenge  of  resource  allocation  for   edge servers in the deployment of BCFL at edge. 
And our contributions can be summarized as follows:
\begin{itemize}
    \item We formulate the resource allocation problem of BCFL in MEC as a multivariate and multi-constraint optimization problem, where the solution is the resource allocation scheme for edge servers to simultaneously handle both the MEC and BCFL tasks without delay.
    \item To solve the above optimization problem progressively, we design two algorithms based on MG-ADMM and MC-ADMM for both the homogeneous and heterogeneous scenarios.
    \item To ensure the convergence of algorithms for more than two variables, we add regularization terms in our proposed algorithms based on MG-ADMM and MC-ADMM, and prove the validity with solid theoretical analysis.
    
    \item We conduct numerous experimental evaluations to prove that the optimization solutions are valid and  our proposed resource allocation schemes are effective.
\end{itemize}

The rest of this paper is organized as follows. We introduce the system model and problem formulation in Section \ref{system}. The MG-ADMM based algorithm to solve the optimization problem in the homogeneous scenario and the MC-ADMM based algorithm for the heterogeneous scenario are displayed in  Sections \ref{gadmm} and \ref{cadmm}, respectively. 
Experimental evaluations are presented in Section \ref{exp}. We discuss the related work in Section \ref{rel}. Finally, we conclude this paper in Section \ref{conc}.


\section{System Model and Problem Formulation}\label{system}
In this section, we will discuss the system model from a general perspective, and then we explore both communication models and computing models of our proposed system respectively. We also analyze the cost model of our proposed resource allocation scheme. 

\subsection{System Model Overview}\label{system_overview}

We consider an edge server $S$ connected with $N$ local mobile devices, denoted as $i\in \left\{1,2,3,\cdots, N \right\}$. Local devices are usually lacking of resources, so they may choose to offload their computing tasks to its nearby edge server $S$. In this way, $S$ can provide necessary resources to 
help local devices finish their offloading tasks in the MEC scenario. At the same time, there are multiple edge servers connecting via blockchain network, where they conduct federated learning (FL) to form a blockchain-based FL (BCFL) service provider. In other words, edge servers will be responsible for not only providing computing services to local devices, but also maintaining the BCFL system at the same time. The topology of our considered system is shown in Fig. \ref{fig_sys}.
\begin{figure}[htpt]
\centering
\includegraphics[width=0.45\textwidth]{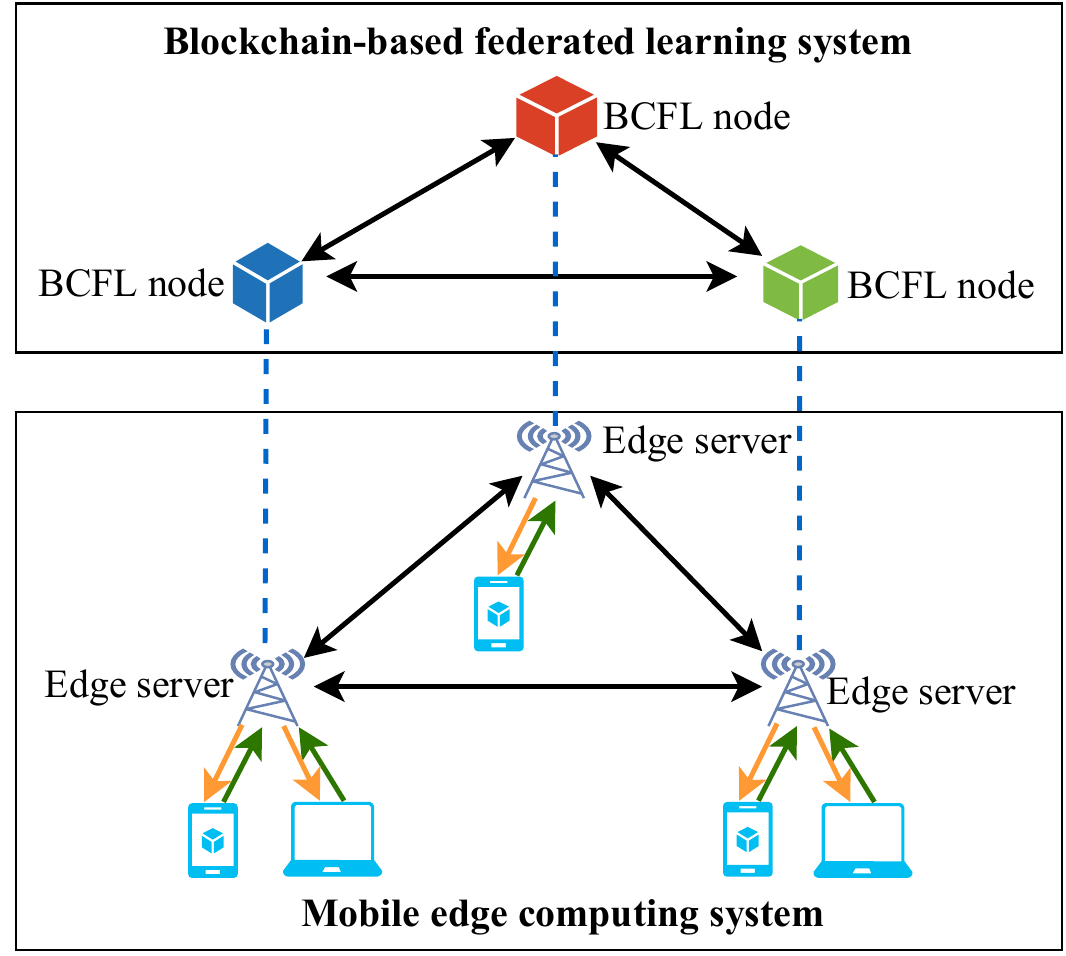}\caption{The topology of BCFL in MEC.}
\label{fig_sys}
\end{figure}

In the MEC, local device $i$ first transmits an offloading request $R_i(D_i, T_i)$ to server $S$, where  $D_i$ is the data size of its task and $T_i$ is the time constraint of this task to be finished. Once $S$ accepts tasks, local devices transmit their data to $S$.
As for the BCFL system, according to \cite{wang2021blockchain}, we consider a fully coupled BCFL which runs FL on a consortium blockchain. First, the clients of FL, i.e., edge servers, train the local models using local data which may be collected from other devices or by themselves, and then they also work as blockchain nodes to generate new blocks which contain the local model updates and the newly updated global model of FL. 
For simplicity, we treat the FL and the blockchain jobs together as the BCFL task, consuming computing, communication, and storage resources. 

Generally, $S$ has limited computing capacity and communication bandwidth, which can be denoted as $F$ and $B$, respectively. Given that the tasks in both the MEC and BCFL are usually time-sensitive, 
simultaneously computing the offloading tasks for lower-layer mobile devices and maintaining the upper-layer BCFL system without any delay require rigorous and optimal resource allocation at edge servers. 

\subsection{Communication Models}
In this subsection, we would like to model the communication resource consumption for finishing the MEC and BCFL tasks, respectively.

\subsubsection{MEC Task} Communications between any device $i$ and the server $S$ include sending offloading request, sending original data and returning computing results. Since the sizes of the offloading request and computing results are smaller than that of the data, we only consider the transmission of original data from devices to the server. 

According to Shannon Bound, the data transmission rate from local device $i$ to edge server $S$ is defined as
\begin{equation}
    r_i^{comm} (\alpha_i)=\alpha_i B \log 2 (1+\frac{P_i G_i}{\delta^2}),\nonumber
\end{equation}
where $\alpha_i \in (0,1)$ represents the percentage of bandwidth allocated to local devices $i$; $B$ is the maximum bandwidth of server $S$; $P_i$ and $G_i$ are the transmission power and gain from $i$ to $S$, respectively; and $\delta$ is the Gaussian noise during the transmission.

Then, we can calculate the time cost of data transmission from device $i$ to server $S$ as 

\begin{equation}
    T_i^{comm} (\alpha_i) = \frac{D_i}{r_i^{comm}(\alpha_i)},\nonumber
\end{equation}
which indicates that the transmission time cost is correlated to the data size of MEC task.

Also, the data transmission will cost a certain amount of energy, which can be calculated by 
\begin{equation}
    E_i^{comm} (\alpha_i)=P_i T_i^{comm}(\alpha_i),\nonumber
\end{equation}
and the total consumption of transmitting the data from all the local devices to the server is calculated as below:
\begin{equation}
    E_{total}^{comm}(\alpha_i)=\sum_{i=1}^{N} E_i^{comm}(\alpha_i).\nonumber
\end{equation}
\subsubsection{BCFL task} The communications during the BCFL task are composed by sharing updates in the blockchain network and conducting blockchain consensus. For simplicity, here we treat the communication in BCFL as a general work process. Let $\alpha_{bcfl}$ denote the percentage of bandwidth distributed to the BCFL task, and let $P_{bcfl}$ and $G_{bcfl}$ represent the transmission power and gain of the BCFL task respectively. Then, we can calculate the data transmission rate in the BCFL task by
\begin{equation}
    r_{bcfl}^{comm} (\alpha_{bcfl})=\alpha_{bcfl} B \log 2 (1+\frac{P_{bcfl} G_{bcfl}}{\delta^2}).\nonumber
\end{equation}

And the time cost of transmission in the BCFL task is
\begin{equation}
    T_{bcfl}^{comm} (\alpha_{bcfl}) = \frac{\widehat{D_{bcfl}}}{r_{bcfl}^{comm} (\alpha_{bcfl})},\nonumber
\end{equation}
where $\widehat{D_{bcfl}}$ is the size of required transmission data in the BCFL task, which is smaller than the size of the training and mining data for the BCFL task, denoted as $D_{bcfl}$, at server $S$.

The energy consumption of the server for conducting the BCFL task can be calculated as
\begin{equation}
     E_{bcfl}^{comm} (\alpha_{bcfl})=P_{bcfl} T_{bcfl}^{comm} (\alpha_{bcfl}).\nonumber
\end{equation}

\subsection{Computing Models}
In this part, we describe the time and energy consumed by the MEC server to process the MEC and BCFL tasks, respectively.

\subsubsection{MEC Task}
Let $f_i\in (0, F)$ be the CPU cycle frequency allocated to the task of device $i$. First, we define the total CPU cycles used for the task of device $i$ as $\mu_i$, and it can be calculated as $\mu_i=D_i d_i$ with $d_i$ denoting the unit CPU cycle frequency required to process one data sample of the MEC task from device $i$. Then, the computing time can be calculated by
\begin{equation}
    T_{i}^{comp}(f_i)=\frac{\mu_i}{f_i}.\nonumber
\end{equation}

According to \cite{burd1996processor}, the energy cost of computing one single task of device $i$ is
\begin{equation}
E_i^{comp}(f_i)=\gamma \mu_i f_i^2,\nonumber
\end{equation}
where $\gamma$ is the parameter correlated to the architecture of the CPU.
Thus, the total energy consumption of computing the MEC tasks for all devices is calculated by
\begin{equation}
    E_{total}^{comp}(f_i)=\sum_{i=1}^N E_i^{comp}(f_i).\nonumber
\end{equation}
\subsubsection{BCFL Task} Similarly, we define $f_{bcfl} \in (0,F)$ as the CPU cycle frequency allocated to the BCFL task. Let $\mu_{bcfl}=D_{bcfl} d_{bcfl}$ denote the total CPU cycles for processing the BCFL task, where $d_{bcfl}$ means the unit CPU cycle used to process one BCFL data sample. 

Then, we can have the time cost of computing the BCFL task:
\begin{equation}
    T_{bcfl}^{comp}(f_{bcfl})=\frac{\mu_{bcfl}}{f_{bcfl}}.\nonumber
\end{equation}

In this way, the energy cost of computing the BCFL task is calculated as:
\begin{equation}
    E_{bcfl}^{comp}(f_{bcfl})=\gamma \mu_{bcfl} f_{bcfl}^2.\nonumber
\end{equation}

\subsection{Cost Model}
We have discussed the energy consumed by the communication and computation of the MEC and the BCFL tasks. Now we can define the cost model of our proposed resource allocation scheme.
Denoting the total energy cost as $U$, based on the above analysis, we know that $U$ is composed of the transmission cost and the computing cost. Then, we have
\begin{align}
    U(\alpha_i, \alpha_{bcfl},f_i,f_{bcfl})=&E_{total}^{comm}(\alpha_i)+E_{bcfl}^{comm}(\alpha_{bcfl})\notag\\
    &+E_{total}^{comp}(f_i)+E_{bcfl}^{comp}(f_{bcfl}).
    \label{utility}
\end{align}

\subsection{Problem Formulation}
The purpose of our resource allocation mechanism is to allow the edge server to handle both the MEC and BCFL tasks satisfying resource and time constraints with the minimum cost. The edge server should make the decisions about how many CPU cycle frequencies and how much bandwidth should be allocated to each task.  Technically, the optimal resource allocation decisions need to consider minimizing the total energy consumption of the edge server. 
Thus, we can formulate the decision making challenge of resource allocation into an optimization problem as below:
\begin{align}
\textbf{P1}: \mathop{\arg\min}_{\alpha_i, \alpha_{bcfl},f_i, f_{bcfl}}&: U \notag\\
\text{s.t.}:\textbf{C1}&:T_{bcfl}^{comm}+T_{bcfl}^{comp}\leq T_{bcfl},\notag   \\
\textbf{C2}&:T_i^{comm}+T_i^{comp}\leq T_i,\notag \\
\textbf{C3}&:\alpha_{bcfl}+\sum_{i=1}^N \alpha_i \leq 1, \notag \\
\textbf{C4}&:f_{bcfl}+\sum_{i=1}^N f_i \leq F,\notag  \\
\textbf{C5}&:D_{bcfl} + \widehat{D_{bcfl}}+ \sum_{i=1}^N D_i \leq D,\notag \\
\textbf{C6}&:f_i,f_{bcfl}\in (0,F), \alpha_i,\alpha_{bcfl}\in (0,1),\notag \\
& \quad   i\in \left\{1,2,\cdots,N\right\}\notag,
\end{align}
where \textbf{C1} and \textbf{C2} guarantee that the server can finish the BCFL task and MEC task on time; \textbf{C3} and \textbf{C4} ensure that the communication and computing resources allocated to each task are not out of the maximum capacities of the server; \textbf{C5} means that the total data size of all the tasks running on the server cannot exceed its maximum storage capacity, denoted as $D$; 
\textbf{C6} clarifies the ranges of all variables. 

\begin{theorem}
The optimization objective function $U(\alpha_i, \alpha_{bcfl},f_i,f_{bcfl})$ is convex.
\label{theorem_con}
\end{theorem}

\begin{proof}
The Hessian Matrix of $U$ respect to $\alpha_i, \alpha_{bcfl}, f_i, f_{bcfl}$ is given by:
\begin{equation}
\setlength{\arraycolsep}{0.5pt}
H_1=\left(                 
  \begin{array}{cccc}   
    \frac{2D_i p_i log(2)}{B\alpha_i^3 \ln(\frac{G_iP_i}{\delta^2}+1)} & 0 & 0 & 0\\  
    0 & \frac{2D_{bcfl} p_{bcfl} log(2)}{B\alpha_{bcfl}^3 \ln(\frac{G_{bcfl}P_{bcfl}}{\delta^2}+1)} & 0 & 0\\  
    0 & 0& 2N\gamma\mu_i & 0\\  
    0 & 0&0&  2\gamma\mu_{bcfl}\\  
 \end{array}
\right).       \nonumber          
\end{equation}

The eigenvalues of matrix $H_1$ are:

\begin{align}
    V_1=\begin{bmatrix}
 \\2\gamma \mu_{bcfl}
 \\2N\gamma \mu_i \\
 \\\frac{2D_{bcfl} p_{bcfl} log(2)}{B\alpha_{bcfl}^3 \ln(\frac{G_{bcfl}P_{bcfl}}{\delta^2}+1)}\\
\\\frac{2D_i p_i log(2)}{B\alpha_i^3 \ln(\frac{G_iP_i}{\delta^2}+1)}\notag\\
\notag
\end{bmatrix}.
\end{align}
It can be seen that all elements in vector $V_1$ are positive. So matrix $H_1$ is a positive definite matrix, and thus we can prove that the optimization objective function $U$ is convex.
\end{proof}


However, it is still hard to solve \textbf{P1} even though the objective function is convex due to the following reasons: 
1) there are multiple variables required to be optimized, and they are not fully correlated; 2) there are multiple constraints, making it harder to find the optimal solutions. 
Hence, we need to design solutions for \textbf{P1}, which will be introduced in the following sections.

\section{MG-ADMM based Solution in the Homogeneous\label{gadmm} Situation}\label{s1}
To present our resource allocation scheme in a progressive way, we will give a benchmark solution of \textbf{P1} in the homogeneous situation where all MEC tasks have the same data size and time requirement.
In this case, we start from a simple case of \textbf{P1} in this section, where an equal distribution strategy is considered to allocate resources to all local devices, including both the bandwidth and CPU cycle frequencies. The equal distribution strategy means that the edge server distributes the communication and computing resources to each local device for MEC tasks in an equal way, that is, $\alpha_i$ and $f_i$ are the same for any arbitrary device $i$. 

According to Boyd \textit{et al.} \cite{boyd2011distributed}, the alternating direction method of multipliers (ADMM), combining \textit{dual ascend} and \textit{dual decomposition}, is designed to solve problems which are multivariate, separable and convex.
We will solve \textbf{P1} with the equal distribution strategy based on the modified general ADMM (MG-ADMM) method, which is derived from the basic form of ADMM. 
In the following, we first introduce MG-ADMM and reformulate the problem based on the MG-ADMM algorithm, and then we explain how we solve \textbf{P1}. 




\subsection{Brief Introduction to MG-ADMM}
First, we introduce G-ADMM as the basis of MG-ADMM.
According to Boyd \textit{et al.} \cite{boyd2011distributed}, G-ADMM tries to solve the following problem:
\begin{align}
    \mathop{\arg\min}_{x, y} \quad f(x)+g(z) \notag \\
    s.t: Ax+Bz=c,\notag
\end{align}
where $x\in \mathbb{R}^n$, $z\in \mathbb{R}^m$, $A\in \mathbb{R}^{p*n}$, $B\in \mathbb{R}^{p*m}$, and $c\in \mathbb{R}^{p}$. Functions $f(x)$ and $g(z)$ are convex,  and $x$ and $z$ are two parameters. The objective of G-ADMM is to find the optimal value $p^*$:
\begin{equation}
    p^*=\inf \left\{f(x)+g(z) | Ax+Bz=c\right\}.
    \nonumber
\end{equation}

Then, we can form the augmented Lagrangian as below:
\begin{align}
    \mathcal{L}_\rho(x, z, y) = &f(x)+g(z)+y^T(Ax+Bz-c)\notag\\
    &+\frac{\rho}{2}\left\|Ax+Bz-c\right\|_2^2.
    \nonumber
\end{align}
where $y$ is the Lagrange multiplier, and $\rho>0$ is the penalty parameter. 

We assume that $k\in \left\{1,2,\cdots,K\right\}$ iterations are required to find the optimal value, and the updates of the iterations are:
\begin{equation}
    x^{k+1}:=\mathop{\arg\min} \mathcal{L}_\rho (x,z^{k}, y^{k}),
    \nonumber
\end{equation}

\begin{equation}
    z^{k+1}:=\mathop{\arg\min} \mathcal{L}_\rho (x^{k+1},z, y^{k}),
    \nonumber
\end{equation}

\begin{equation}
    y^{k+1} := y^k+\rho(Ax^{k+1}+Bz^{k+1}-c).
    \nonumber
\end{equation}

It has been proved that when the following two conditions are satisfied, the ADMM algorithm can converge: 1) The functions $f: \mathbb{R}^n \rightarrow \mathbb{R}\cup(+\infty)$ and $g: \mathbb{R}^m \rightarrow \mathbb{R}\cup(+\infty)$ are closed, proper, and convex; 2) the augmented Lagrangian $\mathcal{L}_\rho(x,z,y)$ has a saddle point.

The above G-ADMM algorithm is the basic form, which is effective to solve the problem which has 2-block (i.e., two separable functions). However, when we need to solve the problem with more than two separable functions, implementing G-ADMM directly can't guarantee convergence.

To handle this issue, He \textit{et al.} \cite{he2015splitting} propose a novel operator splitting method. In this paper, we term it as MG-ADMM. Let's take 3-block separable minimization problem as the example to describe MG-ADMM when the number of blocks is more then 2.

The form of 3-block separable minimization problem is:
\begin{equation}
    \mathop{min}\left\{ f(x)+g(z)+h(y) | Ax+Bz+Ch=b\right\}.
    \nonumber
\end{equation}

Then the Lagrangian function is:
\begin{align}
    \mathcal{L}_\rho(x,z,y,\lambda )=&f(x)+g(z)+h(y)\notag\\
    &+\lambda^T(Ax+Bz+Ch-b)\notag\\
    &+\left\|Ax+Bz+Cy-b\right\|_2^2.
    \nonumber
\end{align}

And the updates of iterations are:
\begin{equation}
    x^{k+1}:=\mathop{\arg\min} \left\{\mathcal{L}_\rho^\beta (x,z^{k}, y^{k}, \lambda^k)\right\},
    \nonumber
\end{equation}

\begin{equation}
    z^{k+1}:=\mathop{\arg\min}\left\{ \mathcal{L}_\rho^\beta (x^{k+1},z, y^{k}, \lambda^k)+\frac{\rho}{2}\beta \left\|B(z-z^k)\right\|_2^2\right\},
    \nonumber
\end{equation}

\begin{equation}
    y^{k+1}:=\mathop{\arg\min}\left\{ \mathcal{L}_\rho^\beta (x^{k+1},z^{k}, y, \lambda^k)+\frac{\rho}{2}\beta \left\|C(y-y^k)\right\|_2^2\right\},
    \nonumber
\end{equation}

\begin{equation}
    \lambda^{k+1} := \lambda^k -\beta(Ax^{k+1}+By^{k+1}+Cz^{k+1}-b),
    \nonumber
\end{equation}
where $\beta>1$ is the penalty parameter.

\subsection{Problem Reformulation based on MG-ADMM}\label{convexity}


In the homogeneous scenario, the edge server distributes the same amount of resources, denoted as $\alpha^*$ and $f^*$, to each local device. 
As for the energy cost of computing, it is the sum of all devices' costs: 
\begin{equation}
    E_{total}^{comp}(f^*)=\sum_{i=1}^N \gamma \mu_i f_i^2=N \gamma \mu_i f^{*2}.\nonumber
\end{equation}
And the communication cost of the MEC tasks is 
\begin{align}
    E_{total}^{comm}(\alpha^*)=&\sum_{i=1}^N P_{i} T_{i}^{comm}\notag=N P_{i} \frac{D_{i}}{\alpha^* B \log 2 (1+\frac{P_i G_i}{\delta^2})}.\notag
\end{align}

Thus, we can rewrite $U$ as:
\begin{align}
U'(\alpha^*, \alpha_{bcfl},f^*,f_{bcfl})
&=E_{total}^{comm}(\alpha^*)+E_{bcfl}^{comm}(\alpha_{bcfl})\notag\\
&+E_{total}^{comp}(f^*)+E_{bcfl}^{comp}(f_{bcfl}).
\label{utility_2}
\end{align}

Besides, the offloading time costs of communication and computing are:
\begin{equation}
    \widehat{T_{i}^{comp}} (f^*)=\frac{\mu_i}{f^*},\nonumber
\end{equation}
\begin{equation}
    \widehat{T_{i}^{comm}} (\alpha^*)=\frac{D_i}{\alpha^* B \log 2 (1+\frac{P_i G_i}{\delta^2})}.\nonumber
\end{equation}

Based on the above analysis, in the case of homogeneous situation, we need to determine four variables, i.e., $\alpha^*,\alpha_{bcfl}, f^*$ and $f_{bcfl}$.  
We can easily prove that $U'$ is convex based on Theorem \ref{theorem_con}. 
So we apply MG-ADMM to optimize 
$U'$ and derive the optimal variables.
In this way, we can reformulate \textbf{P1} as below:
\begin{align}
\textbf{P2:} \mathop{\arg\min}_{\alpha^*, \alpha_{bcfl},f^*, f_{bcfl}}&:   U'\notag\\
\text{s.t.}:\textbf{C1}&,  \textbf{C5} ~\mathrm{in}~ \textbf{P1} 
, \notag   \\
\textbf{C2}&:\widehat{T_i^{comm}}+\widehat{T_i^{comp}}\leq T_i,\notag \\
\textbf{C3}&:\alpha_{bcfl}+N \alpha^* \leq 1, \notag \\
\textbf{C4}&: f_{bcfl}+N f^* \leq F,\notag  \\
\textbf{C6}&: f^*,f_{bcfl}\in(0,F),\notag \\
&\quad \alpha^*,\alpha_{bcfl}\in (0,1),\notag \\
&\quad i\in \left\{1,2,\cdots,N\right\}\notag.
\end{align}

\subsection{Solution based on MG-ADMM}

First, we form the augmented Lagrangian of \textbf{P2} as follows:

\begin{align}
    \mathcal{L}_1&=\mathcal{L}(\alpha^*,\alpha_{bcfl},f^*,f_{bcfl}, \lambda_1,\lambda_2,\lambda_3,\lambda_4,\lambda_5)\notag\\
    &=U'+\lambda_1(T_{bcfl}^{comm}+T_{bcfl}^{comp}-T_{bcfl})\notag\\
    &+\lambda_2(\widehat{T_i^{comm}}+\widehat{T_i^{comp}}-T_i)\notag\\
    &+\lambda_3(\alpha_{bcfl}+N\alpha^*-1)\notag\\
    &+\lambda_4(f_{bcfl}+N f^*-F)\notag\\
    &+\lambda_5(D_{bcfl}+\widehat{D_{bcfl}}+\sum_{i}^N D_i -D)\notag\\
    &+\frac{\rho}{2}\left\|T_{bcfl}^{comm}+T_{bcfl}^{comp}-T_{bcfl}\right\|_2^2\notag\\
    &+\frac{\rho}{2}\left\|\widehat{T_i^{comm}}+\widehat{T_i^{comp}}-T_i\right\|_2^2\notag\\
    &+\frac{\rho}{2}\left\|\alpha_{bcfl}+N\alpha^*-1\right\|_2^2\notag\\
    &+\frac{\rho}{2}\left\|f_{bcfl}+N f^*-1\right\|_2^2\notag\\
    &+\frac{\rho}{2}\left\|D_{bcfl}+\widehat{D_{bcfl}}+\sum_{i}^N D_i -D\right\|_2^2,\notag
    \nonumber
\end{align}
where $\lambda_m>0 $ with $m\in\left\{1,2,3,4,5\right\}$ is the augmented Lagrange multiplier, and $\rho>0$ is the penalty parameter.

\begin{theorem}
The augmented Lagrangian of \textbf{P2}, i.e., $\mathcal{L}_1$, has a saddle point. 
\label{the1}
\end{theorem}

\begin{proof} The Hessian matrix of $\mathcal{L}_1$ is shown in (\ref{hes}).
\begin{small}
 \begin{figure*}[!t] 
 	\centering
 	\begin{equation}	    
    H_2=\begin{pmatrix}
  \frac{N^2\rho}{1-N\alpha_i-\alpha_{bcfl}}+\frac{D_i \log(2)(2\lambda_2+2P_i+\rho)}{\alpha_i^3 B \ln (1+\frac{P_i G_i}{\delta^2})}& 0 & 0 &0 \\
  0& \frac{3\rho }{8\,{\left(1-\alpha _{\mathrm{bcfl}}-N\,\alpha _{i}\right)}}+\frac{D_{bcfl} \log(2)(2\lambda_1+2P_{bcfl}+\rho)}{\alpha_{bcfl}^3 B \ln (1+\frac{P_{bcfl} G_{bcfl}}{\delta^2})} &0  &0 \\
  0&  0& 2N\gamma \mu_i &0 \\
  0& 0 & 0 &2\gamma \mu_{bcfl}
    \end{pmatrix}
    \label{hes}
 	\end{equation}
 \end{figure*}
 \end{small}

Then we calculate the eigenvalues of matrix $H_2$ as
\begin{align}
    V_2=\begin{bmatrix}
 \\\frac{D_i \log(2)(2\lambda_2+2P_i+\rho)}{\alpha_i^3 B \ln (1+\frac{P_i G_i}{\delta^2})}-\frac{N^2\rho}{1-N\alpha_i-\alpha_{bcfl}}
 \\\frac{D_{bcfl} \log(2)(2\lambda_1+2P_{bcfl}+\rho)}{\alpha_{bcfl}^3 B \ln (1+\frac{P_{bcfl} G_{bcfl}}{\delta^2})}-\frac{3\rho }{8(1-\alpha _{bcfl}-N\alpha _{i})}
 \\2\gamma \mu_{bcfl}
\\2N\gamma \mu_i
\end{bmatrix}.\notag
\end{align}

In vector $V_2$, it is clear that $2\gamma \mu_{bcfl}$ and $2N\gamma \mu_i$ are positive. As for $\frac{D_i \log(2)(2\lambda_2+2P_i+\rho)}{\alpha_i^3 B \ln (1+\frac{P_i G_i}{\delta^2})}-\frac{N^2\rho}{1-N\alpha_i-\alpha_{bcfl}}$ and $\frac{D_{bcfl} \log(2)(2\lambda_1+2P_{bcfl}+\rho)}{\alpha_{bcfl}^3 B \ln (1+\frac{P_{bcfl} G_{bcfl}}{\delta^2})}-\frac{3\rho }{8(1-\alpha _{bcfl}-N\alpha _{i})}$, we cannot know whether they are non-negative. If we let $\frac{D_i \log(2)(2\lambda_2+2P_i+\rho)}{\alpha_i^3 B \ln (1+\frac{P_i G_i}{\delta^2})}-\frac{N^2\rho}{1-N\alpha_i-\alpha_{bcfl}}<0$, then we have $\frac{N^2\rho}{1-N\alpha_i-\alpha_{bcfl}}>\frac{D_i \log(2)(2\lambda_2+2P_i+\rho)}{\alpha_i^3 B \ln (1+\frac{P_i G_i}{\delta^2})}$. In other words, if the above condition is satisfied, then we can say that at least one of the elements in vector $V_2$ is negative. In this way, matrix $H_2$ is a positive semi-definite matrix. Thus, $\mathcal{L}_1$ has a saddle point.
\label{aP1}
\end{proof}


Let $k\in \{1,2,\cdots,K\}$ be the iteration round, and the updates of variables can be expressed as:
\begin{align}
    \alpha^{*k+1}:=&\mathop{\arg\min} \mathcal{L}(\alpha^*,\alpha_{bcfl}^k,f^{*k},f_{bcfl}^k,\lambda_1^k,\lambda_2^k,\lambda_3^k,\lambda_4^k,\lambda_5^k),
    \label{eq_a1}
\end{align}
\begin{align}
    \alpha_{bcfl}^{k+1}:=&\mathop{\arg\min}(
    \mathcal{L}(\alpha^{*k+1},\alpha_{bcfl},f^{*k},f_{bcfl}^k,\notag\\
    &\lambda_1^k,\lambda_2^k,\lambda_3^k,\lambda_4^k,\lambda_5^k)+\frac{\rho}{2}\beta \left\|\alpha_{bcfl}-\alpha_{bcfl}^k\right\|_2^2),
    \label{eq_a2}
\end{align}
\begin{align}
    f^{*k+1}:=&\mathop{\arg\min}(
    \mathcal{L}(\alpha^{*k+1},\alpha_{bcfl}^k,f^*,f_{bcfl}^k,\notag\\
    &\lambda_1^k,\lambda_2^k,\lambda_3^k,\lambda_4^k,\lambda_5^k)+\frac{\rho}{2}\beta \left\|N(f^*-f^{*k})\right\|_2^2),
    \label{eq_f1}
\end{align}
\begin{align}
    f_{bcfl}^{k+1}:=&\mathop{\arg\min}(
    \mathcal{L}(\alpha^{*k+1},\alpha_{bcfl}^k,f^{*k},f_{bcfl},\notag\\
    &\lambda_1^k,\lambda_2^k,\lambda_3^k,\lambda_4^k,\lambda_5^k)+\frac{\rho}{2}\beta \left\|f_{bcfl}-f_{bcfl}^k\right\|_2^2),
    \label{eq_f2}
\end{align}
where $\beta>0$ is the penalty parameter.

The updates of augmented Lagrange multipliers are: 
\begin{align}
    \lambda_1^{k+1}:=\lambda_1^k-\beta(T_{bcfl}^{comm}(\alpha_{bcfl}^{k+1})+T_{bcfl}^{comp}(f_{bcfl}^{k+1})-T_{bcfl}),
    \label{l_1}
\end{align}
\begin{align}
    \lambda_2^{k+1}:=\lambda_2^k-\beta(\widehat{T_i^{comm}}(\alpha^{*k+1})+\widehat{T_i^{comp}}(f^{*k+1})-T_i),
    \label{l_2}
\end{align}
\begin{align}
    \lambda_3^{k+1}:=\lambda_3^k-\beta(\alpha_{bcfl}+N\alpha^*-1),
    \label{l_3}
\end{align}
\begin{align}
    \lambda_4^{k+1}:=\lambda_4^k-\beta(f_{bcfl}+Nf^*-F),
    \label{l_4}
\end{align}
\begin{align}
    \lambda_5^{k+1}:=\lambda_5^k-\beta(D_{bcfl} + \widehat{D_{bcfl}}+ {\sum_{i=1}^N D_i}  -D).
    \label{l_5}
\end{align}

Then, we can set the stopping criteria for above iterations: 
\begin{equation}
\left\|\alpha^{*k+1}-\alpha^{*k}\right\|_2^2\leq\psi,
\left\|f^{*k+1}-f^{*k}\right\|_2^2\leq\psi,
\label{eq_s1}
\end{equation}
\begin{equation}
\left\|\alpha_{bcfl}^{k+1}-\alpha_{bcfl}^{k}\right\|_2^2\leq\psi, \left\|\alpha_{bcfl}^{k+1}-\alpha_{bcfl}^{k}\right\|_2^2\leq\psi,
\label{eq_s2}
\end{equation}
where $\psi$ is the predefined threshold \cite{xiong2019cloud}.

Note that (\ref{eq_a1}) to (\ref{eq_f2}) are quadratic optimization problems and can be solved easily. Due to the space limitation, we omit the detailed calculations.


It has been proved that when the following two conditions are satisfied, the MG-ADMM algorithm can converge: 1) the objective function is closed, proper, and convex; and 2) the augmented Lagrangian  has a saddle point.
We have proved that the objective function is convex, and it is also closed and proper. Besides, we have proved that $\mathcal{L}_1$ has a saddle point in \textbf{Theorem} \ref{the1}. Thus, the convergence of $\textbf{P2}$ is guaranteed.

We summarize our proposed solution based on MG-ADMM in Algorithm \ref{al_1}. First, we initialize four variables and five augmented Lagrangian multipliers (Line 1), and then we update the variables and Lagrange multipliers in an iterative process (Lines 2-17). Specifically, we update variables and Lagrange multipliers (Lines 3-11), and calculate the stopping criteria (Line 12). If the termination condition is satisfied, then the objective function is converged (Lines 13-15). In the end, we calculate the optimal value of the objective function, and then all the optimal decisions and the optimal total energy cost are returned (Lines 18-19). 

\begin{algorithm}
\caption{Solution of \textbf{P2} based on MG-ADMM Algorithm} 
\label{al_1}
\begin{algorithmic}[1]
\REQUIRE $P_i$, $D_i$, $N$, $G_i$, $G_{bcfl}$, $\delta$, $\psi$, $G_{bcfl}$, $F$, $\gamma$, $d_{bcfl}$, $\rho$, $P_{bcfl}$, $D_{bcfl}$, $k$, $T_i$, $T_{bcfl}$, $\beta$, $\lambda_1$, $\lambda_2$, $\lambda_3$, $\lambda_4$, $\lambda_5$
\ENSURE $\alpha^*$, $\alpha_{bcfl}$, $f^*$, $f_{bcfl}$, $U'$
\STATE Initialize $\alpha^*$, $\alpha_{bcfl}$, $f^*$, $f_{bcfl}$, $\lambda_1$, $\lambda_2$, $\lambda_3$, $\lambda_4$, $\lambda_5$
\WHILE{Convergence $\neq$ True }
\STATE $\alpha^{*k+1}$ $\leftarrow$ find the optimal value of (\ref{eq_a1})
\STATE $\alpha_{bcfl}^{k+1}$ $\leftarrow$ find the optimal value of (\ref{eq_a2})
\STATE $f^{*k+1}$ $\leftarrow$ find the optimal value of (\ref{eq_f1})
\STATE $f_{bcfl}^{k+1}$ $\leftarrow$ find the optimal value of (\ref{eq_f2})
\STATE $\lambda_1^{k+1}$ $\leftarrow$ update (\ref{l_1})
\STATE $\lambda_2^{k+1}$ $\leftarrow$ update (\ref{l_2})
\STATE $\lambda_3^{k+1}$ $\leftarrow$ update (\ref{l_3})
\STATE $\lambda_4^{k+1}$ $\leftarrow$ update (\ref{l_4})
\STATE $\lambda_5^{k+1}$ $\leftarrow$ update (\ref{l_5})
\STATE Calculate (\ref{eq_s1}) and (\ref{eq_s2})
\IF{(\ref{eq_s1}) and (\ref{eq_s2}) are satisfied}
\STATE Convergence $=$ True
\ENDIF
\STATE $k$  $\leftarrow$  $k+1$
\ENDWHILE
\STATE Calculate $U'$ via (\ref{utility_2})
\RETURN $\alpha^*$, $\alpha_{bcfl}$, $f^*$, $f_{bcfl}$, $U'$
\end{algorithmic}
\end{algorithm}


\section{MC-ADMM based Solution in the Heterogeneous Scenario}\label{cadmm}
In this section, we consider the heterogeneous scenario with diverse MEC requests from local devices. To this end, we need to apply an on-demand resource allocation strategy. 
That is to say, we have to determine the resource allocation decisions for each MEC task, which is more realistic compared to the equal distribution strategy in the homogeneous scenario. Specifically, we calculate $\alpha_i$ and $f_i$ for $i \in \left\{1,2,\cdots,N\right\}$, as well as $\alpha_{bcfl}$ and $f_{bcfl}$. Thus, the optimization problem in this scenario is more practical and complicated.
In the following, we first introduce modified consensus ADMM (MC-ADMM), which is another form of ADMM. Then we formulate \textbf{P1} based on the MC-ADMM algorithm, and the basic idea is to separate the whole optimization task into multiple subtasks which can be resolved in a distributed manner.

\subsection{Brief Introduction to MC-ADMM}
At the beginning, we introduce the C-ADMM, which is one of the ADMM forms. It is designed to solve the following problem:
\begin{align}
    \mathop{\arg\min}_{x} \quad \sum_{i=1}^N f_i(x),
    \nonumber
\end{align}
where $x\in \mathbb{R}^n$ and $f_i: \mathbb{R}^n \rightarrow \mathbb{R}\cup \left\{+\infty\right\}$ are assumed convex. 

The basic idea of C-ADMM is dividing a large scale optimization problem into $N$ subproblems which can be solved in a distributed manner. So, for $\sum_{i=1}^N f_i(x)$, we can rewrite it as:
\begin{align}
    \mathop{\arg\min}_{x} \quad \sum_{i=1}^N f_i(x_i)\notag\\
    s.t.\quad x_i-z=0.\notag\\
    \nonumber
\end{align}
where $z\in \mathbb{R}^n$ is called as auxiliary variable or global variable.

The augmented Lagrangian is:
\begin{align}
    \mathcal{L}(x_1,x_2,\cdots, x_n, z, y)=&\sum_{i=1}^{N}(f_i(x_i)+y_i^T(x_i-z)\notag\\
    &+\frac{\rho}{2}\left\|x_i-z\right\|_2^2),
    \nonumber
\end{align}
where $(x_1,x_2,\cdots,x_n)\in \mathbb{R}^{nN}$.

The updates of parameters are as:
\begin{equation}
    x_i^{k+1}:=\mathop{\arg\min} \left\{\mathcal{L} (f_i(x_i),z^{k}, y_i^{k})\right\},
    \nonumber
\end{equation}
\begin{equation}
    z^{k+1} := \frac{1}{N}\sum_{i=1}^{N}(x_i^{k+1}+\frac{1}{\rho}y_i^k),
    \nonumber
\end{equation}
\begin{equation}
    y_i^{k+1} := y_i^k+\rho(x_i^{k+1}-z^{k+1}).
    \nonumber
\end{equation}

Similar to the MG-ADMM built upon G-ADMM , MC-ADMM is based on C-ADMM by adding regularization terms to the Augmented Lagrangian formula and the variable iteration formulas. Therefore, we omit the detailed formulas of MC-ADMM for brevity.

\subsection{Problem Reformulation based on MC-ADMM}
In the heterogeneous scenario, we have to distribute resources to each  MEC task and the BCFL task, so there are $2N+2$ variables in total. Directly applying the previous MG-ADMM algorithm in this case is not practical since the resource distribution in the heterogeneous situation is much more complicated than the optimization in the homogeneous scenario. Besides, the convergence for $2N+2$ variables in the MG-ADMM algorithm is not guaranteed. Therefore, we resort to the MC-ADMM algorithm,  which can solve the 
large-scale optimization problem in a distributed way. 

Intuitively, allocating the resources to each device is to divide the bandwidth and CPU cycle frequency into $N+1$ parts to find the best decision separately. 
To calculate $\alpha_i$ and $f_i$ for each $i\in\{1,\cdots,N\}$, we first define $\hat{\alpha}$ and $\hat{f}$ as global variables, also called auxiliary variables, 
to assist the distributed optimization. Besides, we have to consider the constraints of \textbf{P1}. For simplicity, we denote the space formed by the constraints related to $\alpha_i$ and $f_i$ (i.e., \textbf{C2}-\textbf{C4} of \textbf{P1}) as $\Omega$, which is the feasible set of local variables $\alpha_i$ and $f_i$. While the other constrains not related to $\alpha_i$ and $f_i$ in \textbf{P1} need to be kept unchanged because they will influence the rest two variables, i.e., $\alpha_{bcfl}$ and $f_{bcfl}$.
Then we can have the reformulated problem as:
\begin{align}
\textbf{P3:} \mathop{\arg\min}_{\alpha_i, \alpha_{bcfl},f_i, f_{bcfl}}&:  U\notag\\
\text{s.t.}: \textbf{C1}&: \alpha_i=\hat{\alpha},\notag   \\
\textbf{C2}&: f_i=\hat{f},\notag \\
\textbf{C3}&: T_{bcfl}^{comm}+T_{bcfl}^{comp} \leq T_{bcfl},\notag \\
\textbf{C4}&: D_{bcfl} + \widehat{D_{bcfl}}+ \sum_{i=1}^N D_i\leq D,\notag \\
\textbf{C5}&:(\alpha_i,f_i)\in \Omega, \alpha_{bcfl}\in(0,1), \notag\\
&\quad f_{bcfl}\in(0,F),\notag \\
&\quad\hat{\alpha}\in (0,1), \hat{f} \in(0,F),\notag\\
&\quad  i\in \left\{1,2,\cdots,N\right\}\notag.
\end{align}

\subsection{Solution based on MC-ADMM}
Here we detail the solution based on MC-ADMM.
First, the augmented Lagrangian form of \textbf{P3} is:
\begin{align}
   \mathcal{L}_2= &\mathcal{L}(\alpha_i, \alpha_{bcfl},f_i,f_{bcfl},\theta_i,\epsilon_i,\eta_1,\eta_2,\hat{\alpha},\hat{f})\notag\\
    &=U+\sum_{i=1}^N \theta_i(\alpha_i-\hat{\alpha})+\sum_{i=1}^N \epsilon_i(f_i-\hat{f})\notag\\
    &+\eta_1(T_{bcfl}^{comm}+T_{bcfl}^{comp}-T_{bcfl})\notag\\
    &+\eta_2(D_{bcfl} + \widehat{D_{bcfl}}+ {\sum_{i=1}^N D_i}  -D)+\frac{\rho}{2}\left\|\alpha_i-\hat{\alpha}\right\|_2^2\notag\\
    &+\frac{\rho}{2}\left\|f_i-\hat{f}\right\|_2^2+\frac{\rho}{2}\left\|T_{bcfl}^{comm}+T_{bcfl}^{comp}-T_{bcfl}\right\|_2^2\notag\\
    &+\left\|D_{bcfl} + \widehat{D_{bcfl}}+ \sum_{i=1}^N D_i-D\right\|_2^2,\notag
    \nonumber
\end{align}
where $\theta_i,\epsilon_i,\eta_1, \eta_2>0$ are augmented Lagrange multipliers.

\begin{theorem}
The augmented Lagrangian of \textbf{P3}, i.e., $\mathcal{L}_2$, has a saddle point. 
\label{the2}
\end{theorem}

\begin{proof}The Hessian matrix of $\mathcal{L}_2$ is shown in (\ref{hes_1}).
 \begin{figure*}[] 
  \centering
  \begin{equation}     
    H_3=\begin{pmatrix}  
    \frac{D_i log(2) (\rho+2NP_i)}{B\alpha_i^3 \log(\frac{G_iP_i}{\delta^2}+1)} & 0 & 0 & 0\\  
    0 & \frac{2D_{bcfl} log(2) (P_{bcfl}+\eta_1)}{B\alpha_{bcfl}^3 \log(\frac{G_{bcfl}P_{bcfl}}{\delta^2}+1)} & 0 & 0\\  
    0 & 0& f_{bcfl}^2 \gamma \mu_{bcfl} + 2 N \gamma \mu_i + \frac{\mu_i \rho}{f_i^3} & 0\\  
    0 & 0&0&  \frac{4\gamma \mu_{bcfl}}{f_{bcfl}^3}\\  
    \end{pmatrix}
    \label{hes_1}
  \end{equation}
 \end{figure*}

Then we calculate the eigenvalues of matrix $H_3$ as
\begin{align}
    V_3=\begin{bmatrix}
 \\f_{bcfl}^2 \gamma \mu_{bcfl} + 2 N \gamma \mu_i + \frac{\mu_i \rho}{f_i^3}
 \\\frac{4\gamma \mu_{bcfl}}{f_{bcfl}^3}
 \\-\frac{D_i log(2) (\rho+2NP_i)}{B\alpha_i^3 \ln(\frac{G_iP_i}{\delta^2}+1)}\\
\\-\frac{2D_{bcfl} log(2) (P_{bcfl}+\eta_1)}{B\alpha_{bcfl}^3 \ln(\frac{G_{bcfl}P_{bcfl}}{\delta^2}+1)}\\
\end{bmatrix}.\notag
\end{align}

Clearly, $f_{bcfl}^2 \gamma \mu_{bcfl} + 2 N \gamma \mu_i + \frac{\mu_i \rho}{f_i^3}>0$ and $\frac{4\gamma \mu_{bcfl}}{f_{bcfl}^3}>0$, while $-\frac{D_i log(2) (\rho+2NP_i)}{B\alpha_i^3 \ln(\frac{G_iP_i}{\delta^2}+1)}<0$ and $-\frac{2D_{bcfl} log(2) (P_{bcfl}+\eta_1)}{B\alpha_{bcfl}^3 \ln(\frac{G_{bcfl}P_{bcfl}}{\delta^2}+1)}<0$. So $H_3$ is a semi-definite matrix, and $\mathcal{L}_2$ has a saddle point.
\label{aP2}
\end{proof}

By applying the method proposed in \cite{he2015splitting}, the updates of local variables (i.e., $\alpha_i$ and $f_i$) are:
\begin{align}
    \left\{\alpha_i^{k+1},f_i^{k+1}\right\}:=\mathop{\arg\min}\mathcal{L}(&\alpha_i, \alpha_{bcfl}^k,f_i,f_{bcfl}^k,\notag\\
    &\theta_i^k,\epsilon_i^k,\eta_1^k,\eta_2^k,\widehat{\alpha^k},\widehat{f^k}),
    \label{l_11}
\end{align}

The updates of $\alpha_{bcfl}$ and $f_{bcfl}$ are:
\begin{align}
    \alpha_{bcfl}^{k+1}:=\mathop{\arg\min}(\mathcal{L}(&\alpha_i^{k+1}, \alpha_{bcfl},f_i^{k+1},f_{bcfl}^{k+1},\notag\\
    &\theta_i^k,\epsilon_i^k,\eta_1^k,\eta_2^k,\widehat{\alpha^k},\widehat{f^k})+\notag\\
    &\frac{\rho}{2}\beta\left\|\alpha_{bcfl}-\alpha_{bcfl}^k\right\|_2^2),
    \label{l_22}
\end{align}
\begin{align}
    f_{bcfl}^{k+1}:=\mathop{\arg\min}(\mathcal{L}(&\alpha_i^{k+1}, \alpha_{bcfl}^{k+1},f_i^{k+1},f_{bcfl},\notag\\
    &\theta_i^k,\epsilon_i^k,\eta_1^k,\eta_2^k,\widehat{\alpha^k},\widehat{f^k})\notag\\
    &+\frac{\rho}{2}\beta\left\|f_{bcfl}-f_{bcfl}^k\right\|_2^2),
    \label{l_33}
\end{align}
where $\rho,\beta>0$ are penalty parameters.

And the updates of global variables are:
\begin{align}
    \widehat{\alpha^{k+1}}:=\frac{1}{N}\sum_{i=1}^N (\alpha_i^{k+1}+\frac{\rho}{2}\theta_i^{k}),
    \label{g_1}
\end{align}
\begin{align}
    \widehat{f^{k+1}}:=\frac{1}{N}\sum_{i=1}^N (f_i^{k+1}+\frac{\rho}{2}\epsilon_i^{k}),
    \label{g_2}
\end{align}

Besides, the updates of augmented Lagrange multipliers are:
\begin{align}
    \theta_i^{k+1}:=\theta_i^{k}+\rho(\alpha_i^{k+1}-\widehat{\alpha^{k+1}}),
    \label{m_1}
\end{align}
\begin{align}
    \epsilon_i^{k+1}:=\epsilon_i^{k}+\rho(f_i^{k+1}-\widehat{f^{k+1}}),
    \label{m_2}
\end{align}
\begin{align}
    \eta_1^{k+1}:=\eta_1^k-\beta(T_{bcfl}^{comm}(\alpha_{bcfl}^{k+1})+T_{bcfl}^{comp}(f_{bcfl}^{k+1})-T_{bcfl}),
    \label{m_3}
\end{align}
\begin{align}
    \eta_2^{k+1}:=\eta_2^k-\beta(D_{bcfl} + \widehat{D_{bcfl}}+ {\sum_{i=1}^N D_i}  -D).
    \label{m_4}
\end{align}

Lastly, the stopping criteria can be set as:
\begin{equation}
\left\|\alpha_i^{k+1}-\widehat{\alpha^{k+1}}\right\|_2^2\leq\psi_{prim},
\left\|f_i^{k+1}-\widehat{f^{k+1}}\right\|_2^2\leq\psi_{prim},
\label{eq_s3}
\end{equation}

\begin{equation}
\left\|\widehat{\alpha^{k+1}}-\widehat{\alpha^{k}}\right\|_2^2\leq\psi_{dual},
\left\|\widehat{f^{k+1}}-\widehat{f^{k}}\right\|_2^2\leq\psi_{dual},
\label{eq_s4}
\end{equation}
where $\psi_{prim}$ and $\psi_{dual}$ are the predefined thresholds \cite{xiong2019cloud}. Besides, (\ref{eq_s2}) should also be included as a stopping criteria.

Even though the forms of $\textbf{P2}$ and $\textbf{P3}$ are different, the proof of the convergence is similar. According to Theorem \ref{theorem_con}, we know that $U$ is convex, and it's clear that $U$ is closed and proper. In addition, the augmented Lagrangian $\mathcal{L}_2$ has a saddle point. So the convergence of $\textbf{P3}$ is guaranteed with the MC-ADMM algorithm.

For reference, we generalize the solution based on MC-ADMM in Algorithm \ref{al_2}. We first initialize local variables, global variables and augmented Lagrangian multipliers (Line 1), and then we calculate the optimal decisions for each MEC task (Lines 2-19). In detail, we keep updating parameters until the objective function is converged (Lines 3-18). Then, we can calculate the optimal total energy cost and return the optimal decisions (Lines 20-21). 

\begin{algorithm}
\caption{Solution of \textbf{P3} based on MC-ADMM Algorithm} 
\label{al_2}
\begin{algorithmic}[1]
\REQUIRE $P_i$, $D_i$, $N$, $G_i$, $G_{bcfl}$, $\delta$, $\psi$, $G_{bcfl}$, $F$, $\gamma$, $d_{bcfl}$, $\rho$, $P_{bcfl}$, $D_{bcfl}$, $k$, $T_i$, $T_{bcfl}$, $\beta$, $\theta_i,\epsilon_i,\eta_1,\eta_2$
\ENSURE $\alpha_i$, $\alpha_{bcfl}$, $f_i$, $f_{bcfl}$, $U$
\STATE Initialize $\alpha_i$, $\alpha_{bcfl}$, $f_i$, $f_{bcfl}$, $\theta_i,\epsilon_i,\eta_1,\eta_2,\hat{\alpha},\hat{f}$

\FOR{ $i \in$ $\left\{1,2,\cdots,N\right\}$}
\WHILE{Convergence $\neq$ True }
\STATE $\alpha_i^{k+1}$, $f_i^{k+1}$ $\leftarrow$ find the optimal values of (\ref{l_11})
\STATE $\widehat{\alpha^{k+1}}$ $\leftarrow$ find the optimal value of (\ref{g_1})
\STATE $\widehat{f^{k+1}}$ $\leftarrow$ find the optimal value of (\ref{g_2})
\STATE $\alpha_{bcfl}^{k+1}$ $\leftarrow$ find the optimal value of (\ref{l_22})
\STATE $f_{bcfl}^{k+1}$ $\leftarrow$ find the optimal value of (\ref{l_33})
\STATE $\theta_i^{k+1}$ $\leftarrow$ update (\ref{m_1})
\STATE $\epsilon_i^{k+1}$ $\leftarrow$ update (\ref{m_2})
\STATE $\eta_1^{k+1}$ $\leftarrow$ update (\ref{m_3})
\STATE $\eta_2^{k+1}$ $\leftarrow$ update (\ref{m_4})
\STATE Calculate (\ref{eq_s2}), (\ref{eq_s3}) and (\ref{eq_s4})
\IF{(\ref{eq_s2}), (\ref{eq_s3}) and (\ref{eq_s4}) are satisfied}
\STATE Convergence $=$ True
\ENDIF
\STATE $k$  $\leftarrow$  $k+1$
\ENDWHILE
\ENDFOR
\STATE Calculate $U$ via (\ref{utility})
\RETURN $\alpha_i$, $\alpha_{bcfl}$, $f_i$, $f_{bcfl}$, $U$
\end{algorithmic}
\end{algorithm}

\section{Experimental Evaluation}\label{exp}
In this section, we design experiments to test the validity and efficiency of our proposed algorithms. We first provide the parameter setting for experiments, then we present and analyze the experimental results. We conduct the experiments using Python 3.8.5 in macOS 11.6 running on Intel i7 processor with 32 GB RAM and 1 TB SSD.

\subsection{Basic Experimental Setting}
We consider a mobile edge computing scenario with one edge server and $10$ local devices. 
For brevity, we provide Table \ref{setting} to detail the basic parameter settings in our experiments. As for the settings of some certain experiments, we will clarify them later. For the augmented Lagrange multipliers, we set them as 1.0 at the beginning.


\begin{table}[H]
\centering
\caption{Basic Experimental Setting}
\setlength{\tabcolsep}{1mm}{
\begin{tabular}{|c|c|c|c|c|}
\hline
$N=10$     & $\beta=0.5$ & $G_i=10$    & $d_i=2$      & $D_i=10$       \\ \hline
$k=100$    & $P_i=2$     & $G_{bcfl}=10$ & $d_{bcfl}=2$ & $D_{bcfl}=10$  \\ \hline
$\rho=0.5$ & $P_{bcfl}=2$  & $F=1000$    & $\delta=0.1$ & $\gamma=0.001$ \\ \hline
$T_i=10$ & $T_{bcfl}=50$  & $\psi=10^{-3}$    & $\psi_{prim}=10^{-3}$ & $\psi_{dual}=10^{-3}$\\ \hline
\end{tabular}}
\label{setting}
\end{table}

Note that we have also tested our proposed algorithms with other parameter settings, while it can be found that the values of data and time related parameters would not affect the changing trends of the experimental results. So, we only report the results for the above parameter settings. Besides, to avoid the statistical bias, we report the averaged results for ten rounds of repeated experiments.



\subsection{Experimental Results}
We design two parts of the experiments: the evaluation of the MG-ADMM based algorithm and the evaluation of the MC-ADMM based algorithm. These two algorithms are designed for different scenarios, i.e., homogeneous and heterogeneous. In the homogeneous scenario, we assume that all the parameters of each MEC task are the same, while in the heterogeneous scenario, we treat each MEC task individually. 
Due to the limitation of space, we only present partial experimental results with importance in this section.

\subsubsection{The Evaluation of the MG-ADMM based Algorithm} We first evaluate the MG-ADMM based algorithm solving \textbf{P2} in the homogeneous scenario, and then we analyze the impacts of the data sizes of both the MEC and the BCFL tasks on the optimal decisions in our resource allocation scheme.

For comparison, we design a \textit{random allocation strategy}, which assigns the bandwidth and CPU cycle frequencies to the MEC and the BCFL tasks in a random way. And we also consider a \textit{fixed allocation strategy}, which determines the resource allocation with fixed values at the beginning. Besides, 
we use the G-ADMM algorithm by setting $\alpha_{bcfl}$ and $f_{bcfl}$ to fixed values as another benchmark solution since setting other variables as constants cannot return converged results.
Via comparing the proposed MG-ADMM based algorithm with these three solutions, we plot the experimental results in Fig. \ref{fig_1}(a). We can see that the MG-ADMM based algorithm can converge after about 80 rounds of iteration, while random strategy cannot converge. In addition, the random and fixed strategies and G-ADMM can only get the total energy cost larger than that of the MG-ADMM algorithm. The results show that the MG-ADMM based algorithm outperforms the other three strategies.

\begin{figure}[h]
\centering
\subfigure[Scheme comparison.]{
\includegraphics[width=0.23\textwidth]{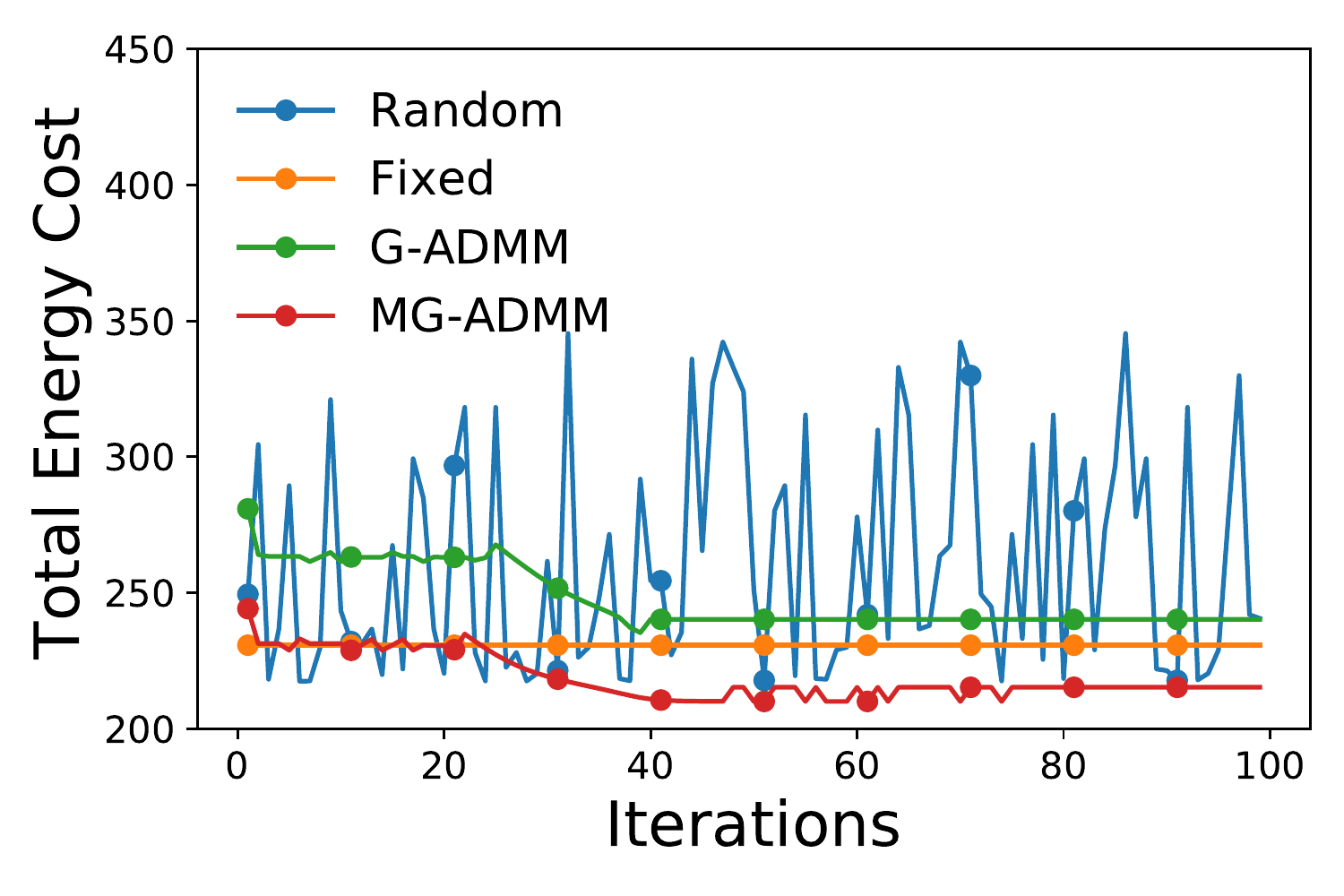}}
\subfigure[Penalty parameter $\rho$.]{
\includegraphics[width=0.23\textwidth]{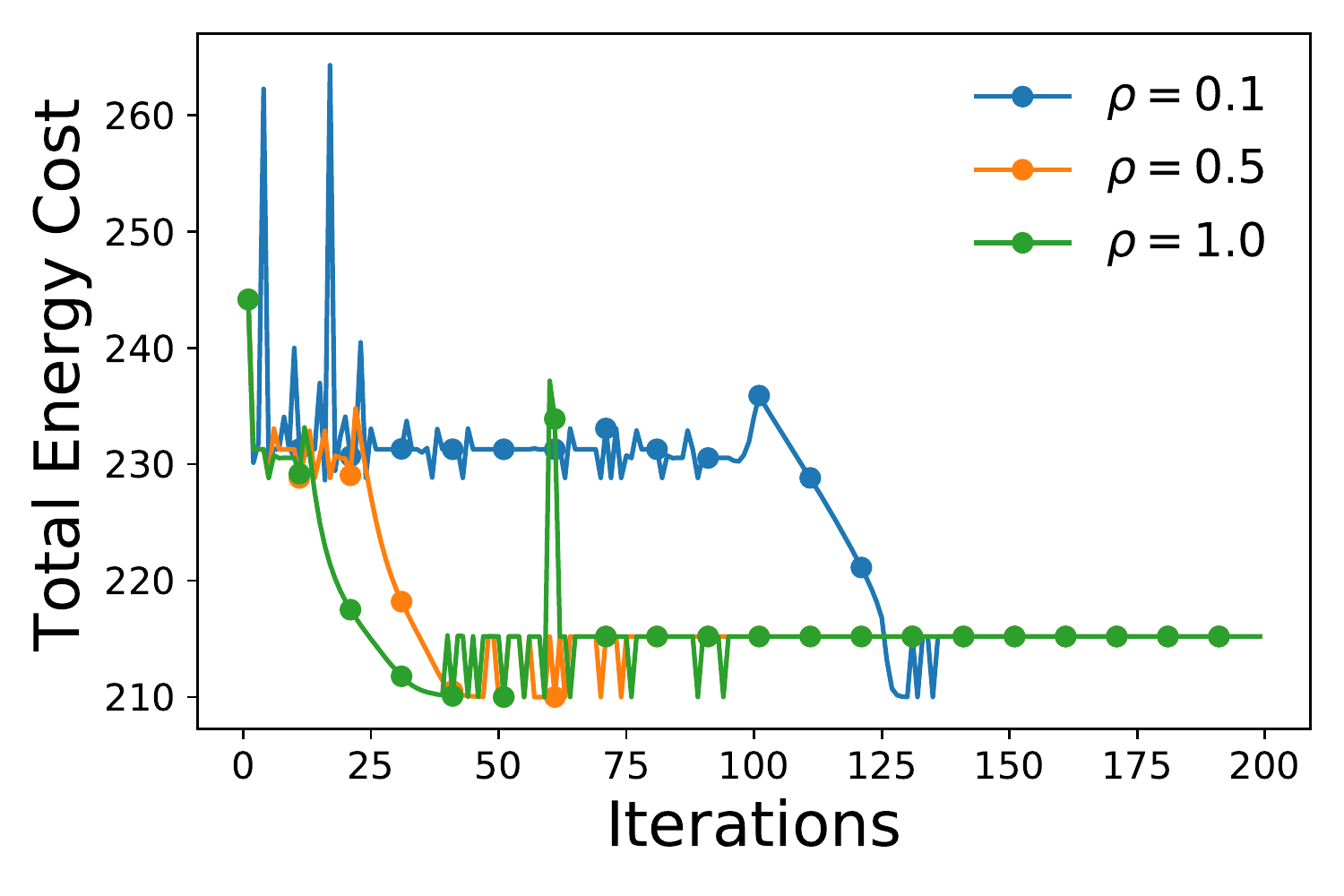}}
\subfigure[Penalty parameter $\beta$.]{
\includegraphics[width=0.23\textwidth]{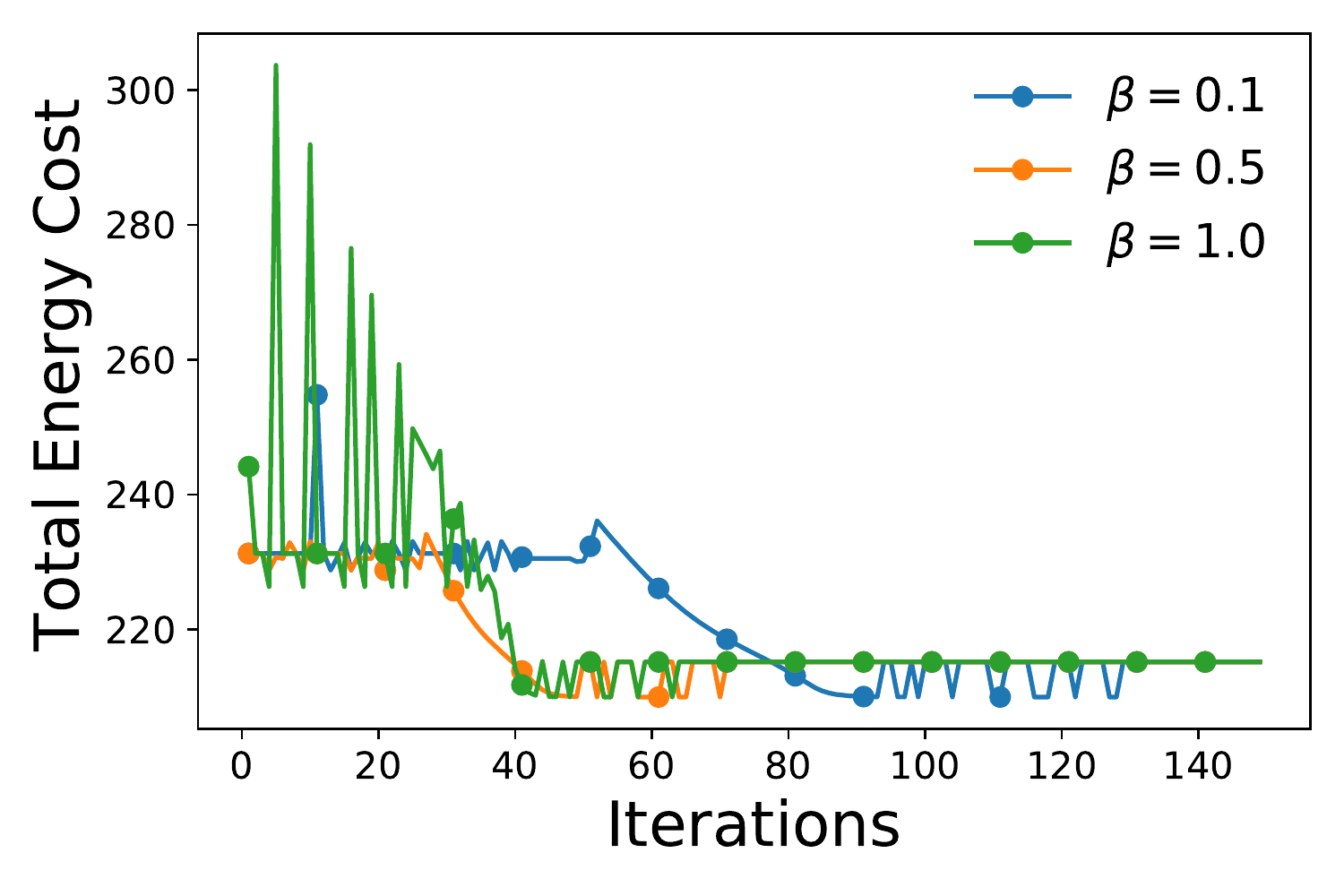}}
\subfigure[The number of devices $N$.]{
\includegraphics[width=0.23\textwidth]{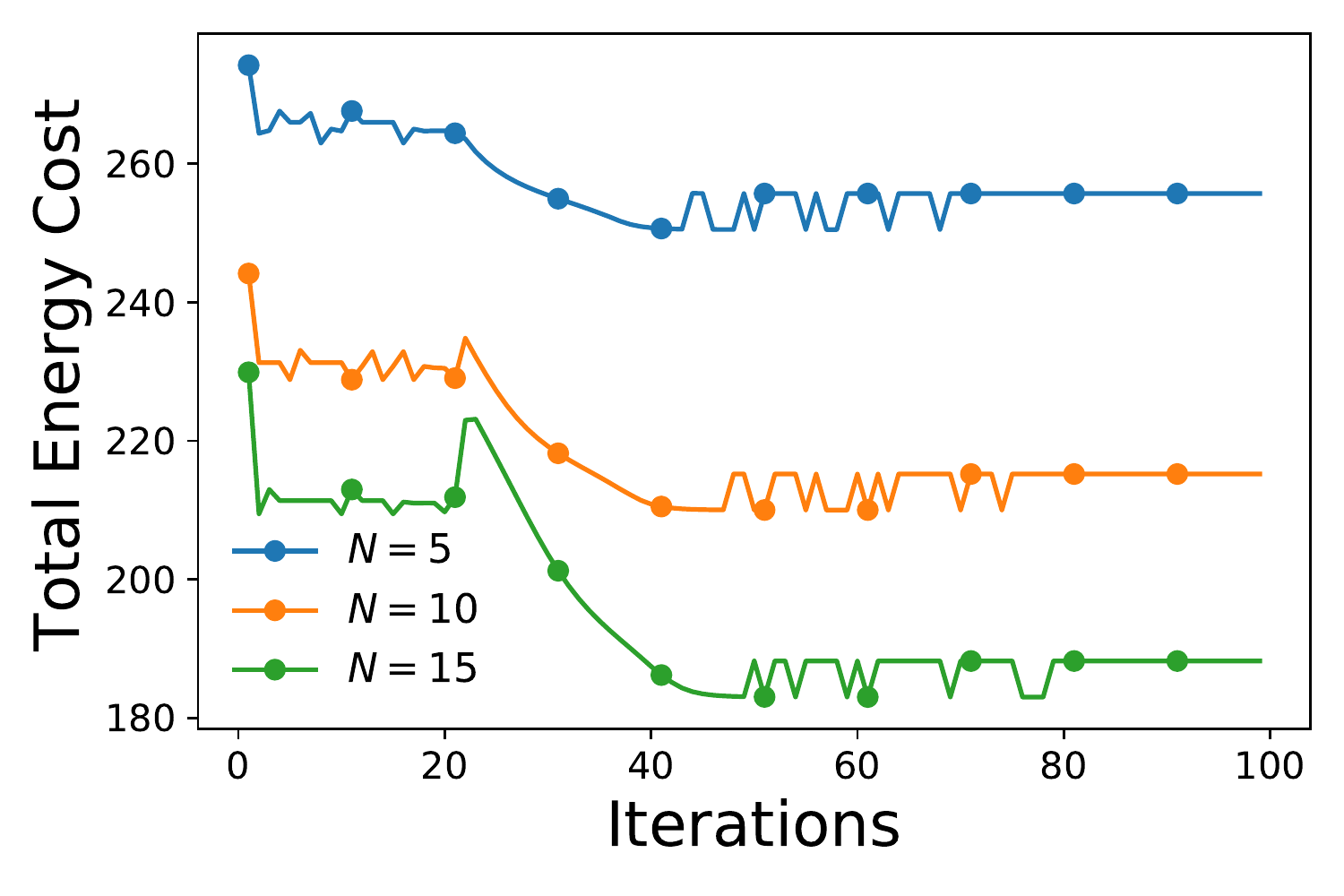}}
\caption{The convergence of the MMG-ADMM based algorithm.}
\label{fig_1}
\end{figure}

\begin{figure}[h]
\centering
\subfigure[Bandwidth.]{
\includegraphics[width=0.23\textwidth]{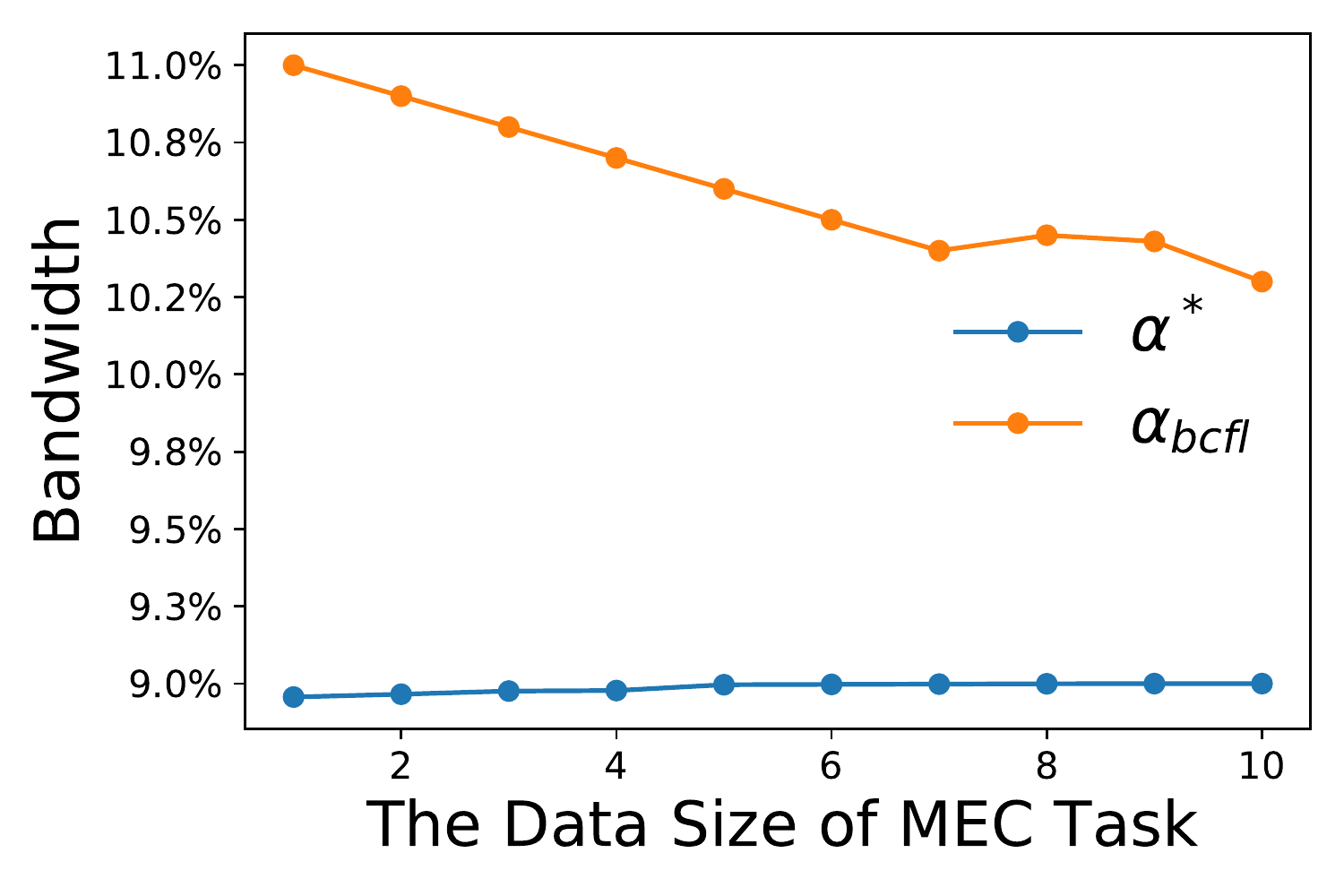}}
\subfigure[CPU cycle frequencies.]{
\includegraphics[width=0.23\textwidth]{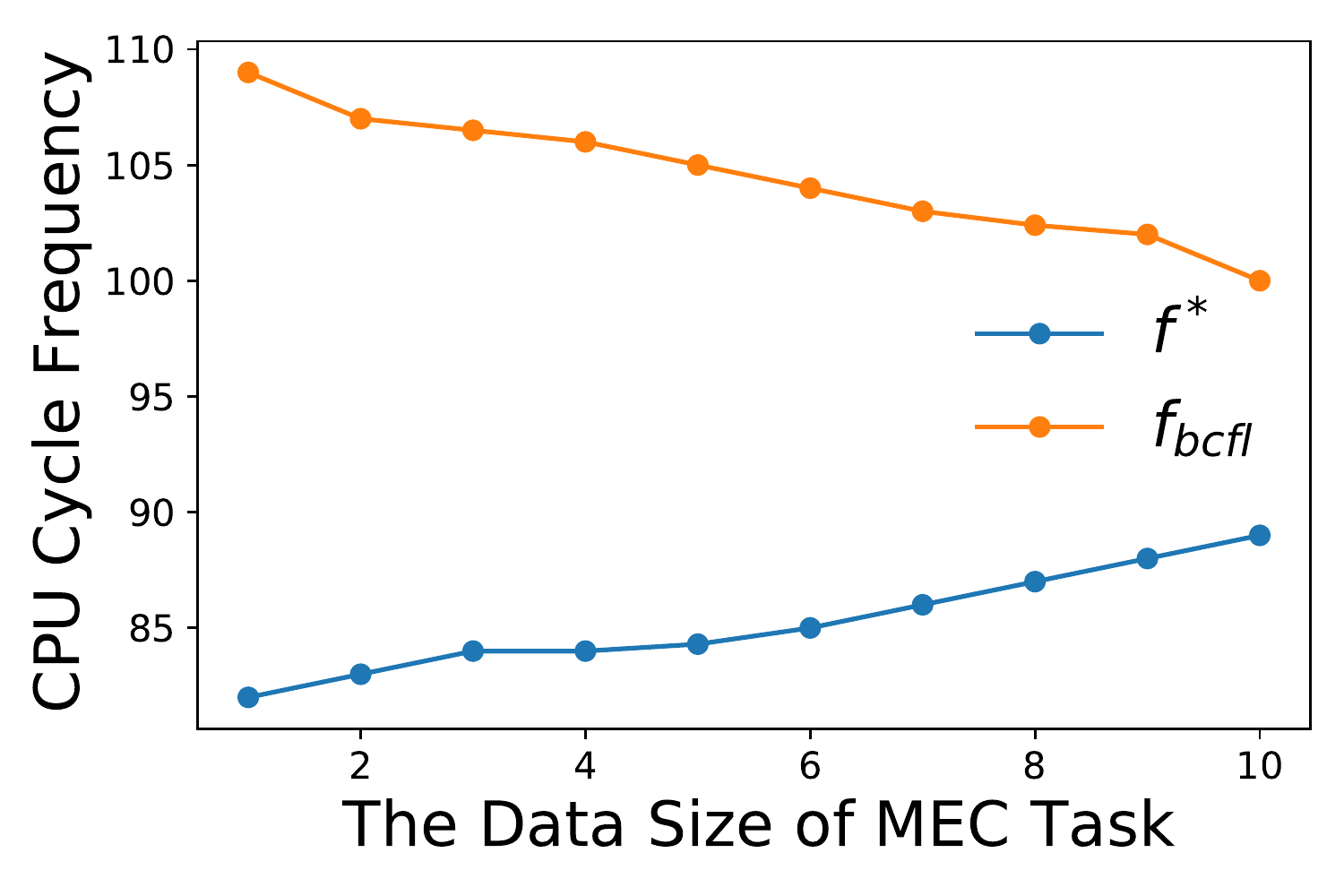}}
\subfigure[Bandwidth.]{
\includegraphics[width=0.23\textwidth]{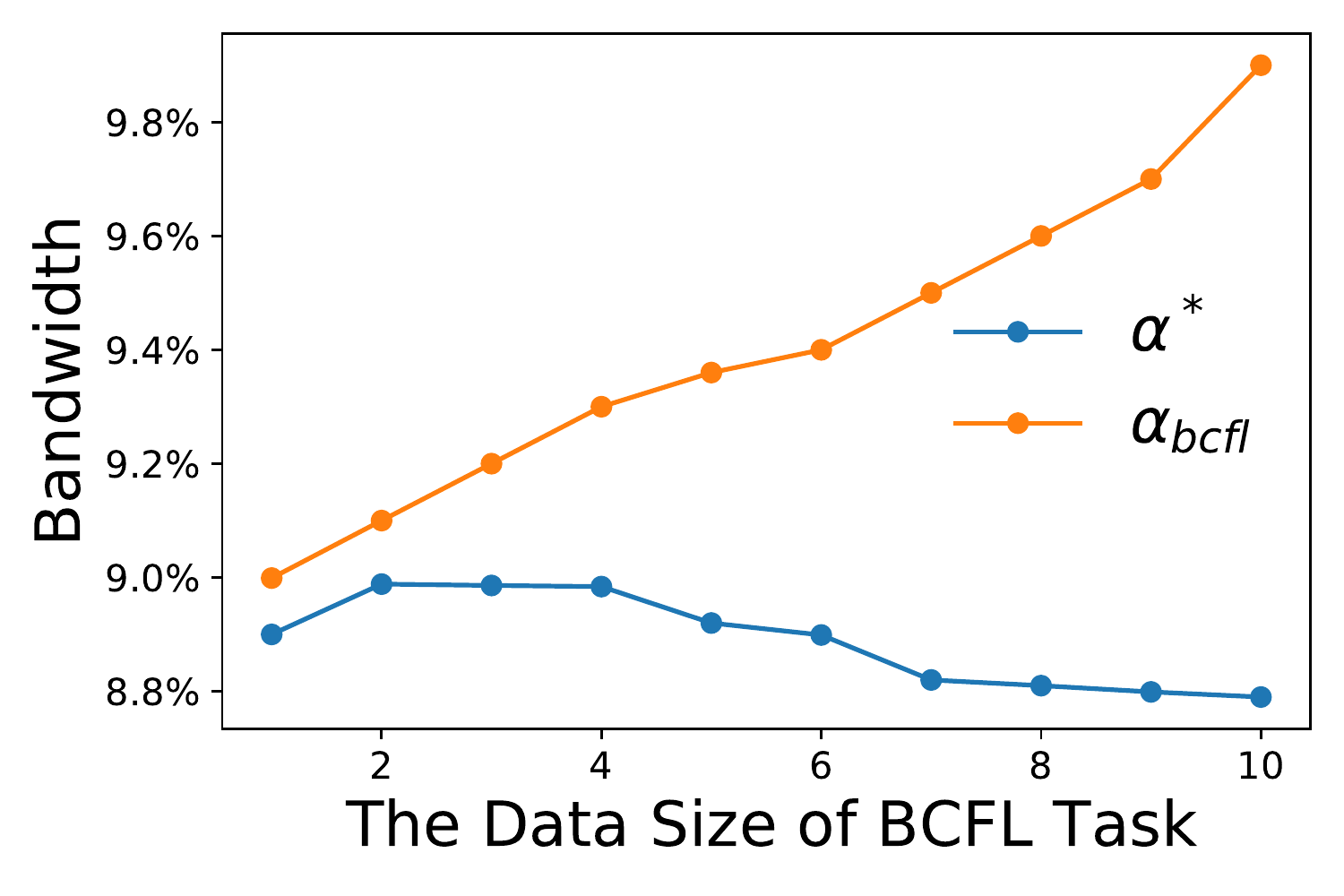}}
\subfigure[CPU cycle frequencies.]{
\includegraphics[width=0.23\textwidth]{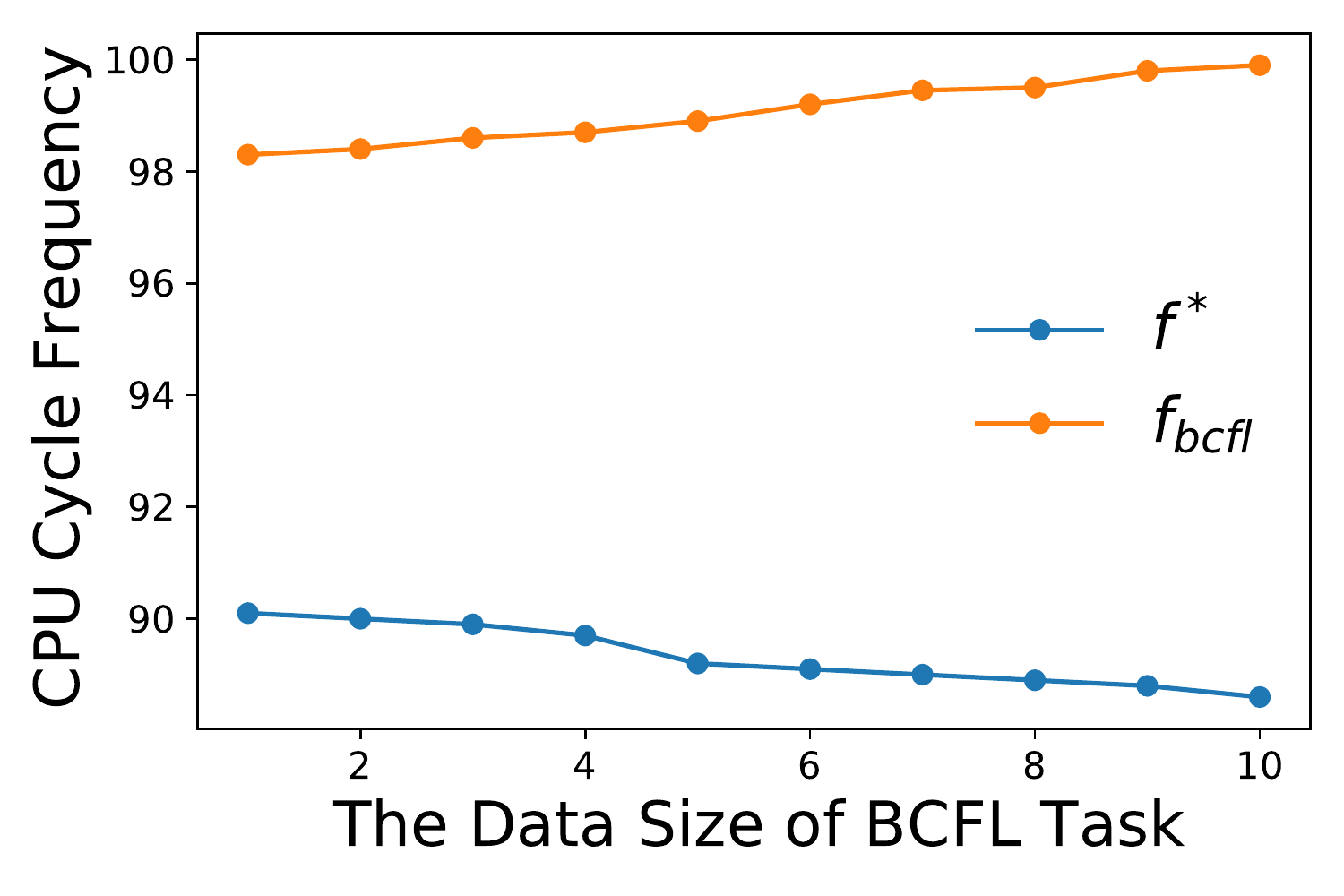}}
\caption{The optimal resource  allocation decisions based on the MG-ADMM algorithm.}
\label{fig_2}
\end{figure}

As the penalty parameters $\rho$ and $\beta$ will influence the convergence speed of the MG-ADMM based algorithm, we set the values of $\rho$ as $\left\{0.1,0.5,1.0\right\}$, and maintain other parameters unchanged. The results in Fig. \ref{fig_1}(b) show that the faster convergence speed comes for a  larger $\rho$. Similarly, we can see from Fig. \ref{fig_1}(c) that the convergence speed will be faster when $\beta$ is larger. 
The reason is that the penalty parameters control the length of the step in each iteration and larger penalty parameters will lead to the greater length of each step, so the convergence speed will be faster. 

To testify the impact of the number of local devices on the convergence of MG-ADMM, we plot experimental results in Fig. \ref{fig_1}(d). We can see that the convergence speed will be slower and the optimal value will be larger when the number of local devices increases, which indicates that it will influence not only the convergence speed but also the optimal value of the total energy cost.  This is because with more devices involved in the MEC tasks, the edge server will cost more energy to work for the tasks, and the optimization problem will be more difficult, so more time will be cost to converge.

In the homogeneous scenario, both the bandwidth and CPU cycle frequencies assigned to each local device are the same, so we only need to calculate four variables, i.e., $\alpha^*, \alpha_{bcfl}, f^*, f_{bcfl}$ for the optimal allocation decisions. In Fig. \ref{fig_2}, for different data sizes of the MEC tasks ($D_i$) and the BCFL task ($D_{bcfl}$), the results show that the data sizes of tasks significantly influence the resource allocation decisions. In Figs. \ref{fig_2}(a) and (b), it can be seen that the larger the data size of each MEC task, the more communication and computing resources allocated to devices and the fewer resources allocated to the BCFL task. Similarly, we can conclude from Figs. \ref{fig_2}(c) and (d) that more resources will be distributed to the BCFL task and fewer resources will be assigned to the MEC tasks if the data size of the BCFL task is larger. The results match the intuition that the larger data size of a task requires more resources in communication and computing.


\begin{figure}[h]
\centering
\subfigure[Strategies comparison.]{
\includegraphics[width=0.23\textwidth]{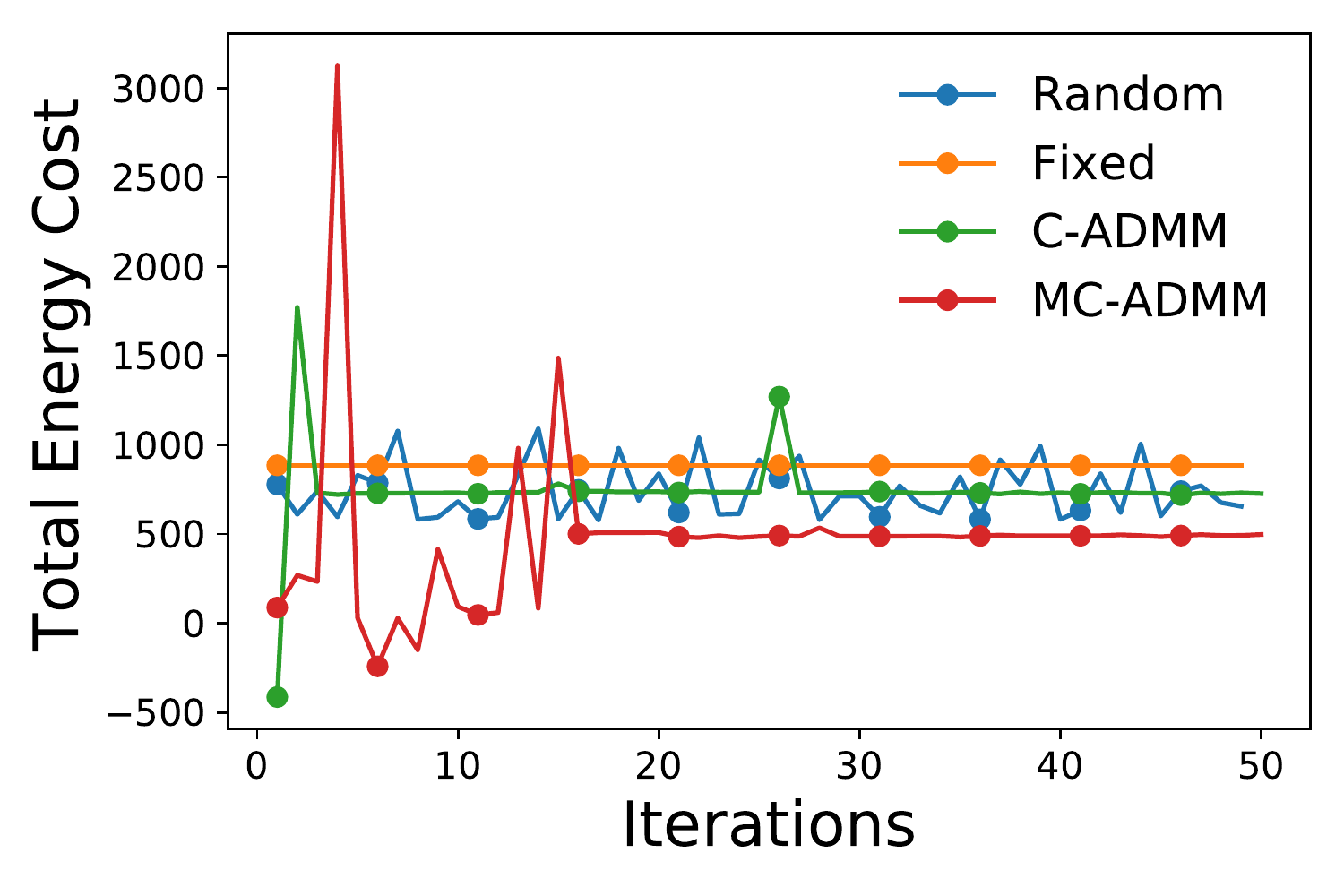}}
\subfigure[The penalty parameter $\rho$.]{
\includegraphics[width=0.23\textwidth]{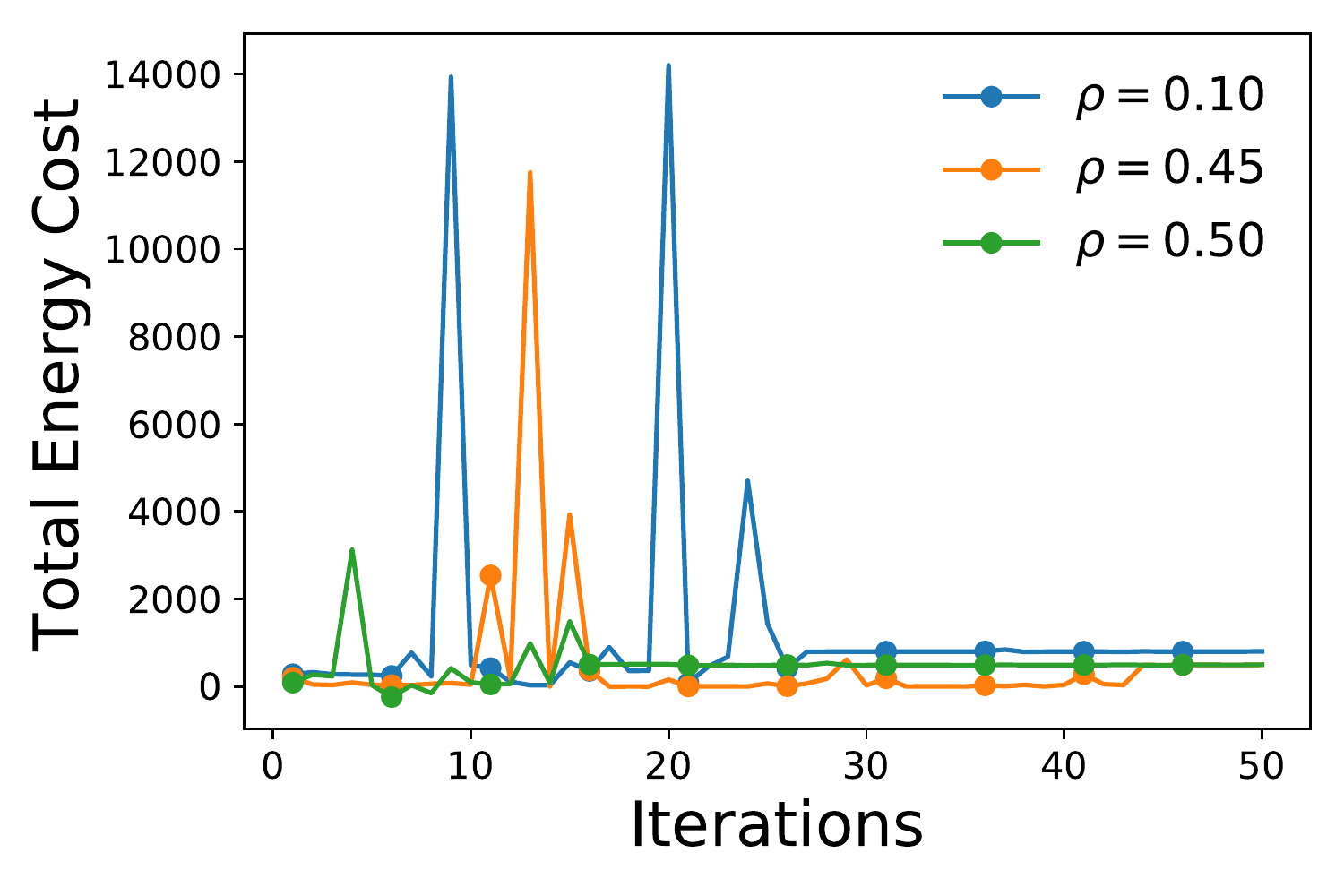}}
\subfigure[The penalty parameter $\beta$.]{
\includegraphics[width=0.23\textwidth]{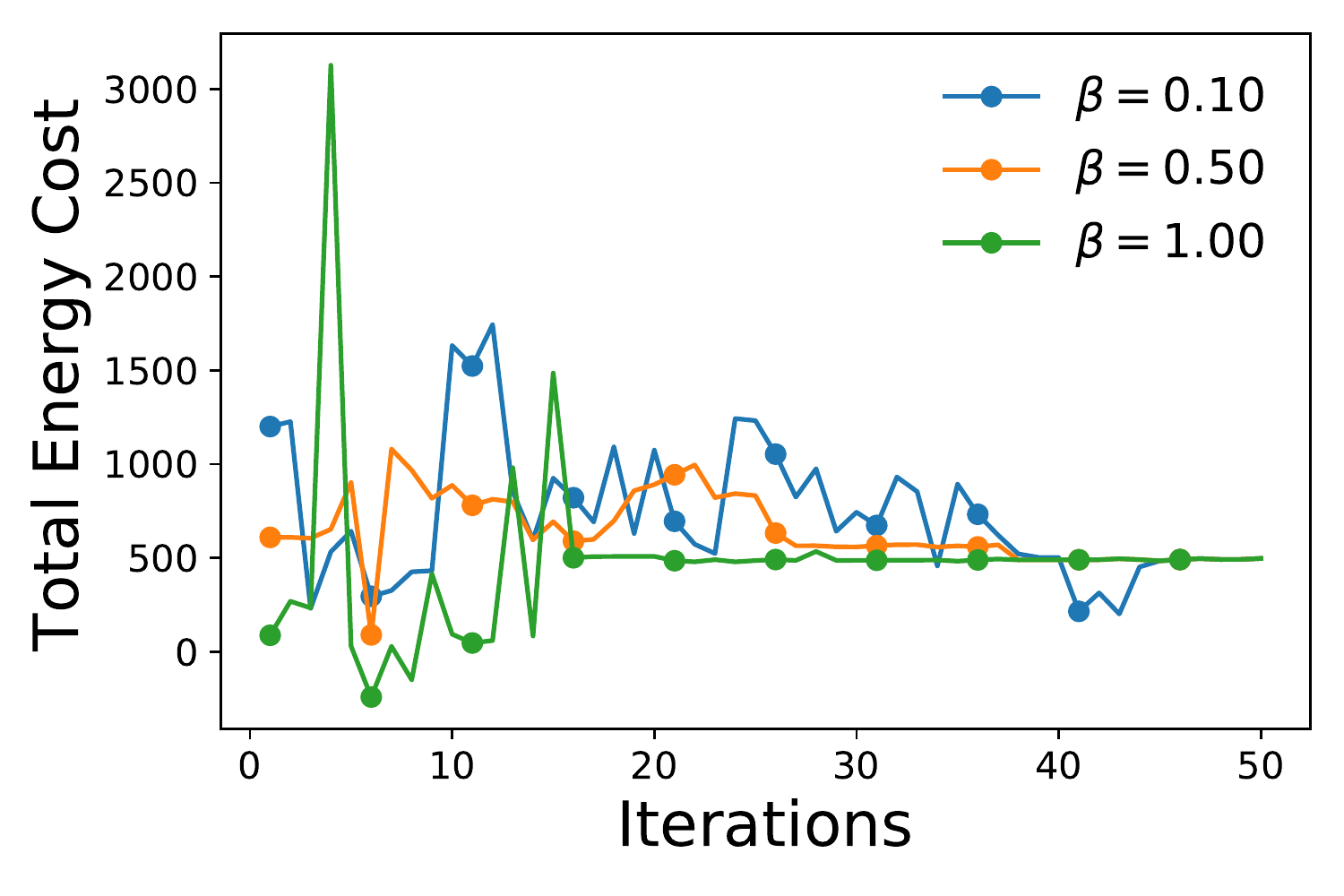}}
\subfigure[The number of devices $N$.]{
\includegraphics[width=0.23\textwidth]{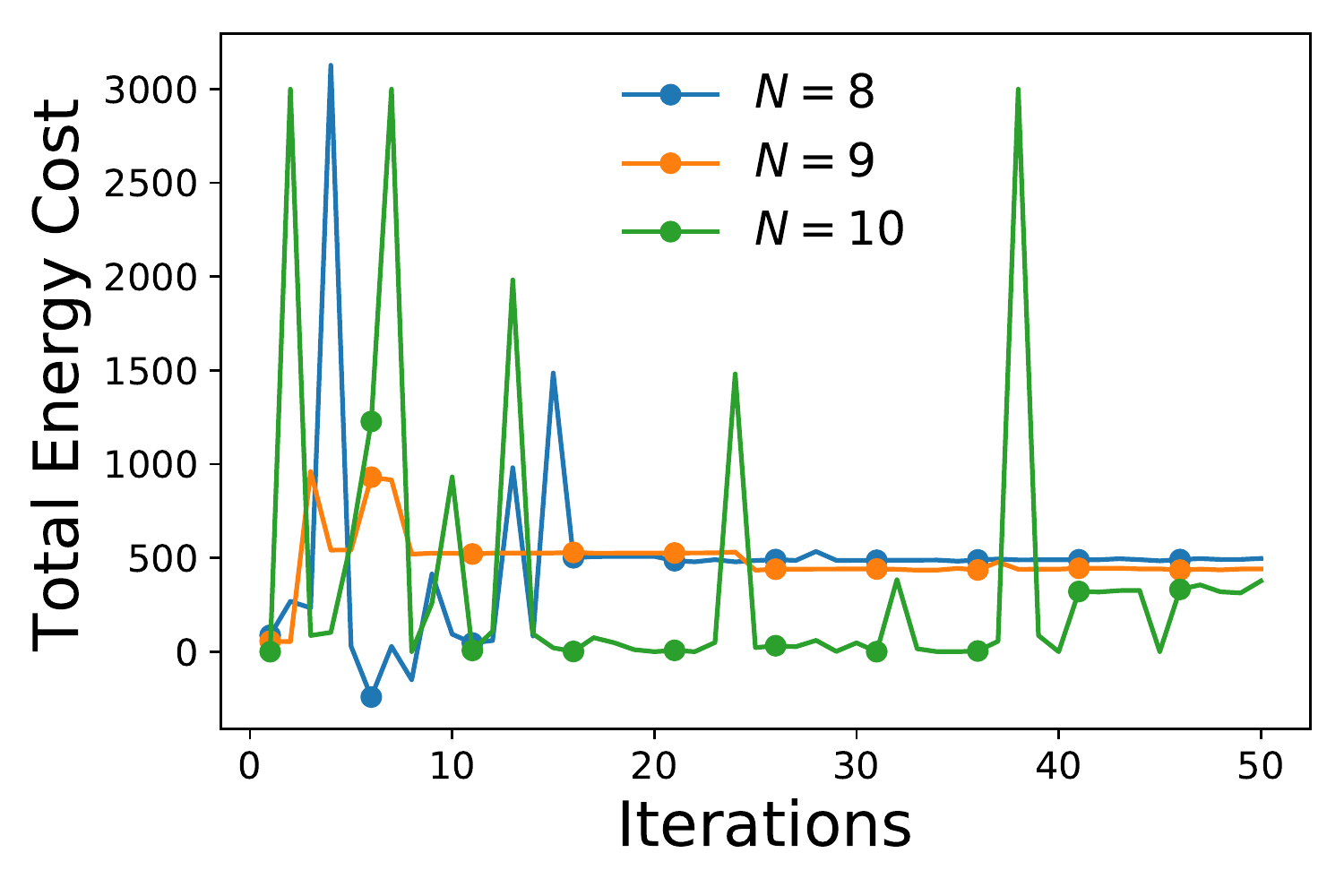}}
\caption{The convergence of the MC-ADMM based algorithm.}
\label{fig_3}
\end{figure}








\subsubsection{The Evaluation of MC-ADMM based Algorithm} 

In this part, the experiments are designed to evaluate the optimization objective of \textbf{P3} from the perspective of convergence and reveal the relationship between the data sizes of tasks and the optimal resource allocation decisions, i.e., the optimization variables in \textbf{P3}. The parameter setting is  $D_i\in\left\{1,2,3,\cdots,10\right\}$ and $T_i\in\left\{1,2,3,\cdots,10\right\}$ with $N=10$, $\rho=0.5$ and $\beta=1.0$, while others are the same with the above experiments.

First, we compare our proposed MC-ADMM based algorithm with the above-mentioned random allocation strategy and fixed allocation strategy and C-ADMM in terms of checking the convergence speed of each strategy. Similar to the setting mentioned above regarding G-ADMM, C-ADMM is implemented by setting $\alpha_{bcfl}$ and $f_{bcfl}$ as constants. 
The results are reported in Fig. \ref{fig_3}(a), which shows that our proposed algorithm performs well in solving \textbf{P3} since it can converge and achieve a lower stable value of the total energy cost than the other three strategies.

Then, we test how penalty parameters $\rho\in\left\{0.10,0.45,0.50\right\}$ and $\beta\in\left\{0.10,0.50,1.00\right\}$ influence the convergence speed. From Figs. \ref{fig_3}(b) and (c), we can know that the larger penalties will cause faster convergence speed. 
What's more, We find that the value of $\rho$ cannot be too large, or the algorithm would not converge.
We also test the influence of the number of local devices ($N\in\left\{8,9,10\right\}$) with the results in Fig. \ref{fig_3}(d) showing that more local devices will lead to more cost and slower convergence speed.

By comparing Figs. \ref{fig_1} and \ref{fig_3}, it can be seen that MG-ADMM requires about 80 rounds to converge, while MC-ADMM only needs less than 50 rounds to reach the stable value, which indicates that the distributed algorithm is more effective.

\begin{figure}[h]
\centering
\subfigure[Bandwidth.]{
\includegraphics[width=0.23\textwidth]{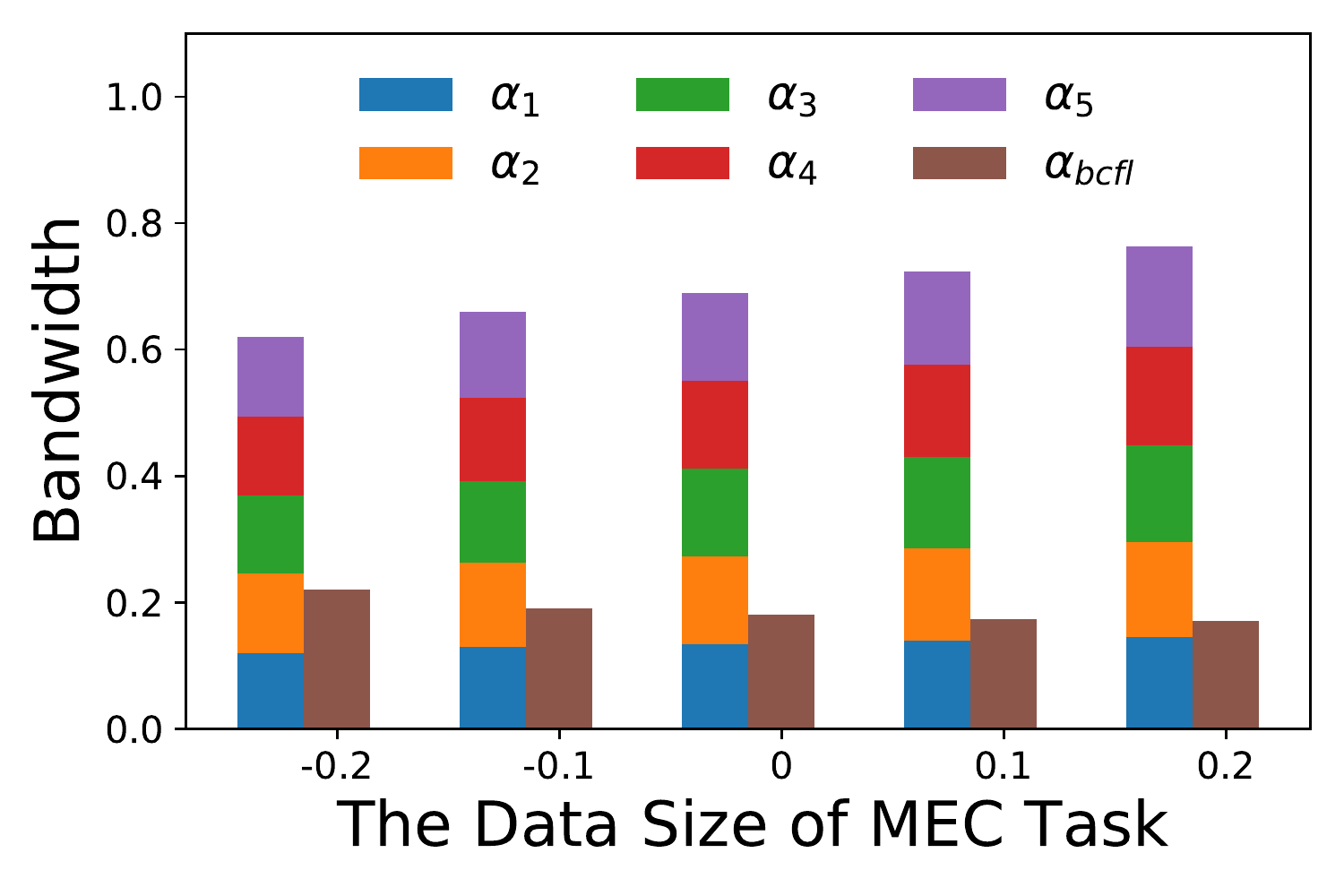}}
\subfigure[CPU cycle frequencies.]{
\includegraphics[width=0.23\textwidth]{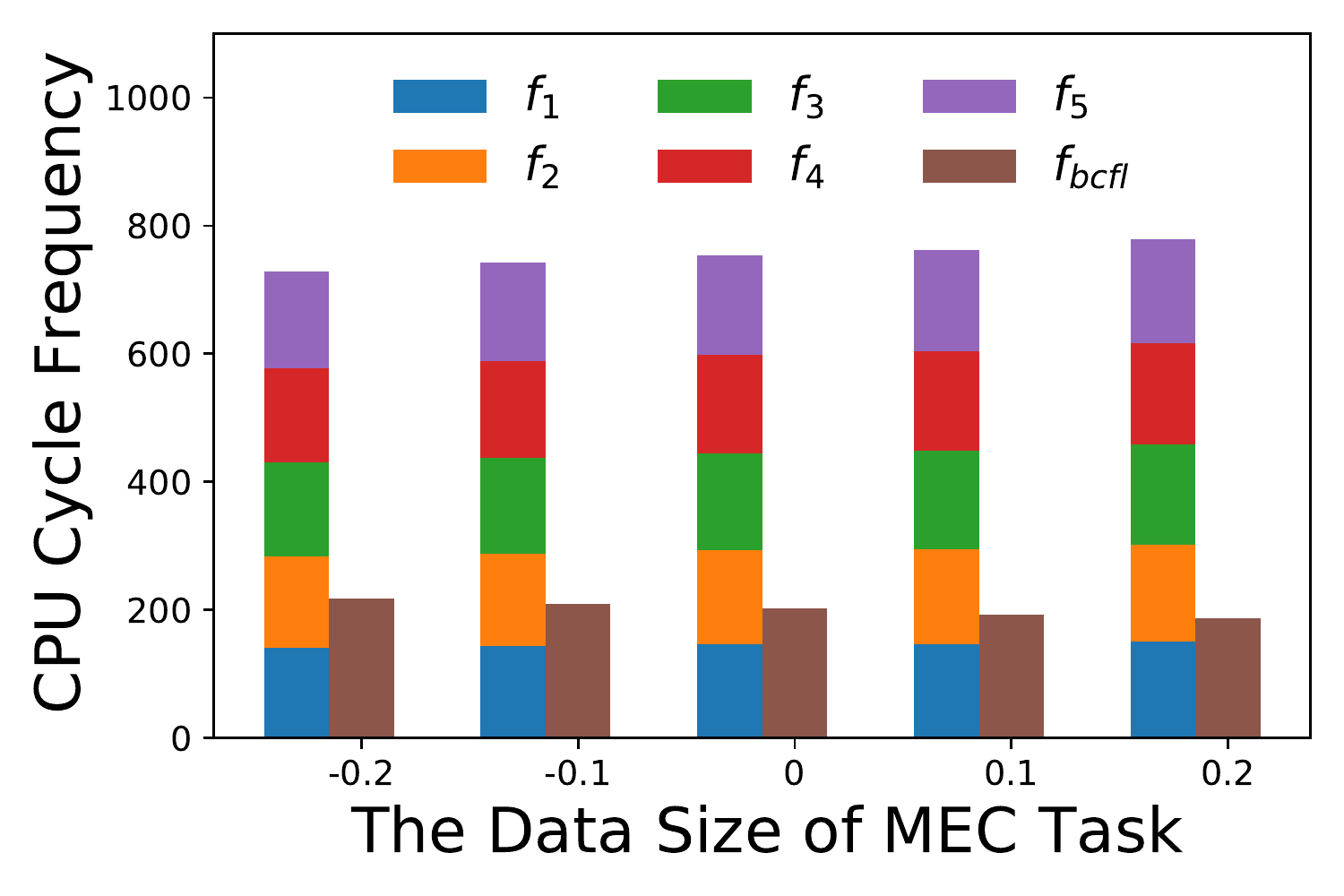}}
\subfigure[Bandwidth.]{
\includegraphics[width=0.23\textwidth]{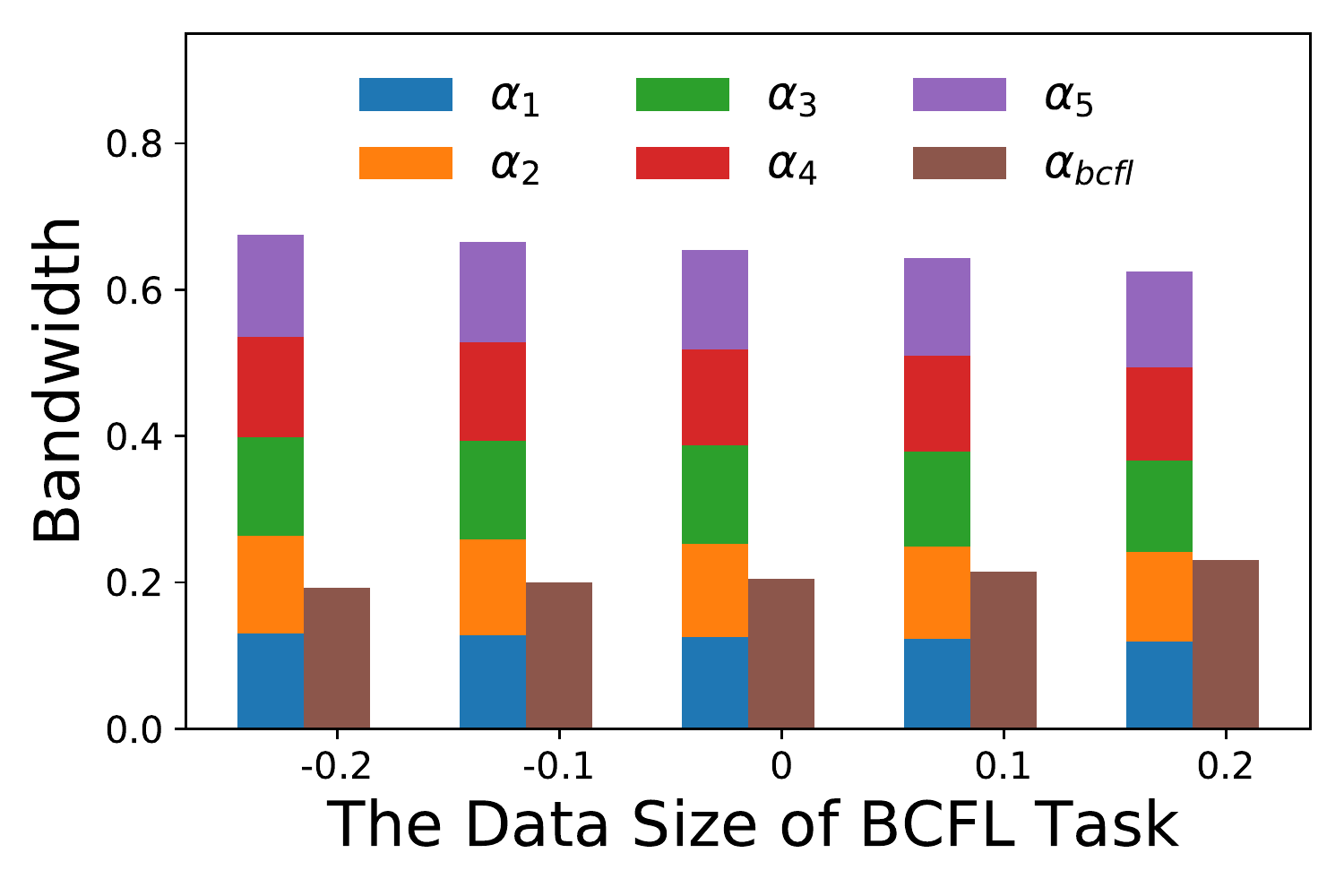}}
\subfigure[CPU cycle frequencies.]{
\includegraphics[width=0.23\textwidth]{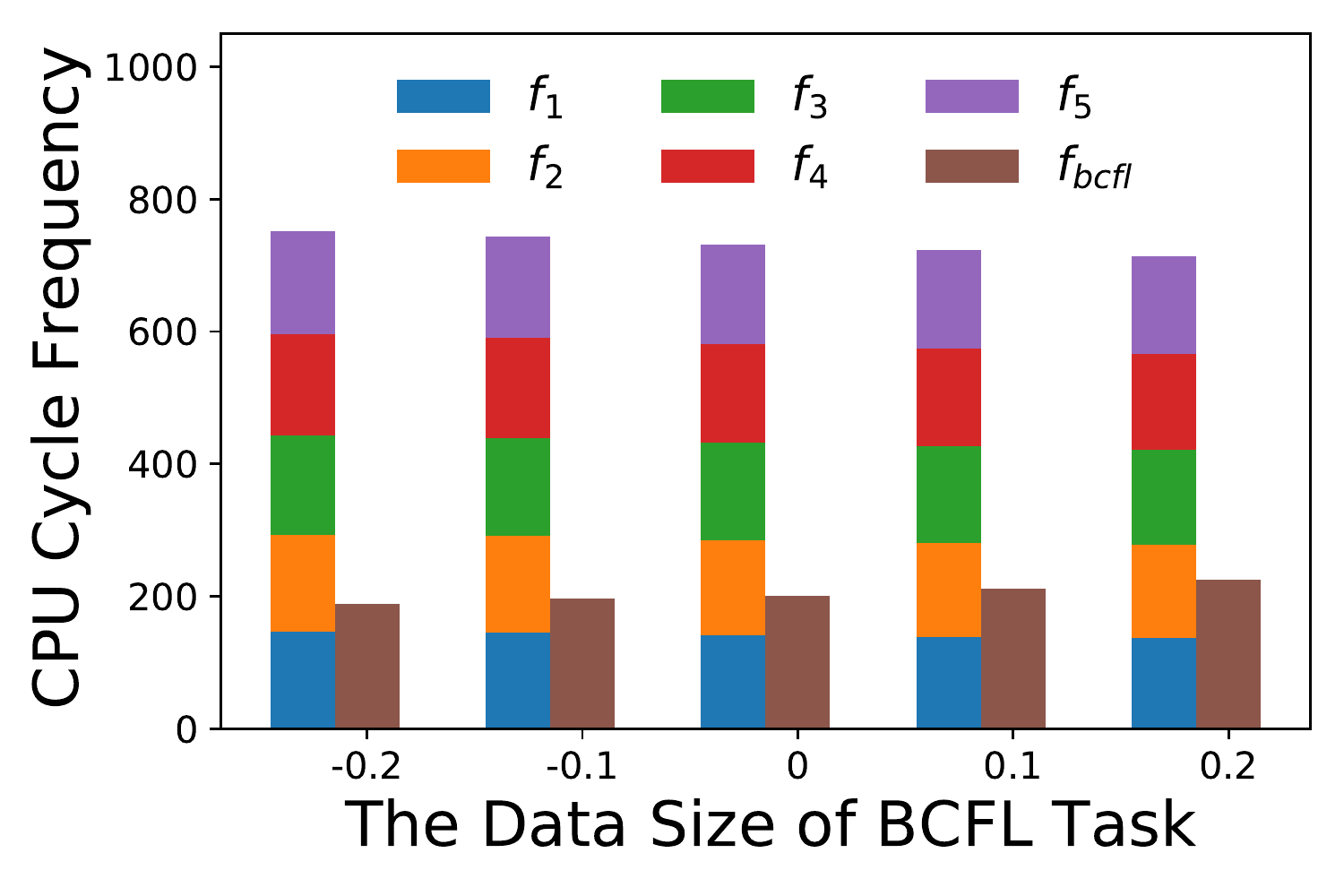}}
\caption{The optimal resource  allocation decisions based on the MC-ADMM algorithm.}
\label{fig_4}
\end{figure}

In \textbf{P3}, we have to determine $\alpha_i$ and $f_i$ for each $i\in\left\{1,2,3,\cdots,N\right\}$, as well as $\alpha_{bcfl}$ and $f_{bcfl}$. Thus, we need to calculate $2N+2$ variables. 
Here, we set $N=5$, and we want to know how the increase and decrease in the sizes of data for the MEC and BCFL tasks affect the optimal decisions. We first let $D_i$ decrease by $10\%$ and $20\%$, and then increase it by $10\%$ and $20\%$. The changes of the percentage are expressed as $\left\{-0.2,-0.1,0,0.1,0.2\right\}$ in Fig. \ref{fig_4}, where $0$ refers to the original data size.
From the results in Figs. \ref{fig_4}(a) and (b), we can see that more resources are allocated to the MEC tasks and fewer resources are distributed to the BCFL task when $D_i$ increases. Conversely, the results in Figs. \ref{fig_4}(c) and (d) show that more resources are assigned to the BCFL task when $D_{bcfl}$ is larger. This is consistent with the changing trends in the homogeneous scenario and can be explained by the same reason that more resources are needed to finish tasks with larger data sizes.







\subsubsection{The Evaluation of Latency}
In an ideal scenario, the MEC server would devote the appropriate resources to task processing based on the decisions obtained by the algorithms we designed. In this part, experiments are conducted to evaluate the latency of processing the MEC and BCFL tasks according to the decisions obtained from our algorithms.

\begin{figure}[h]
\centering
\subfigure[Latency changes with the data size of MEC task.]{
\includegraphics[width=0.23\textwidth]{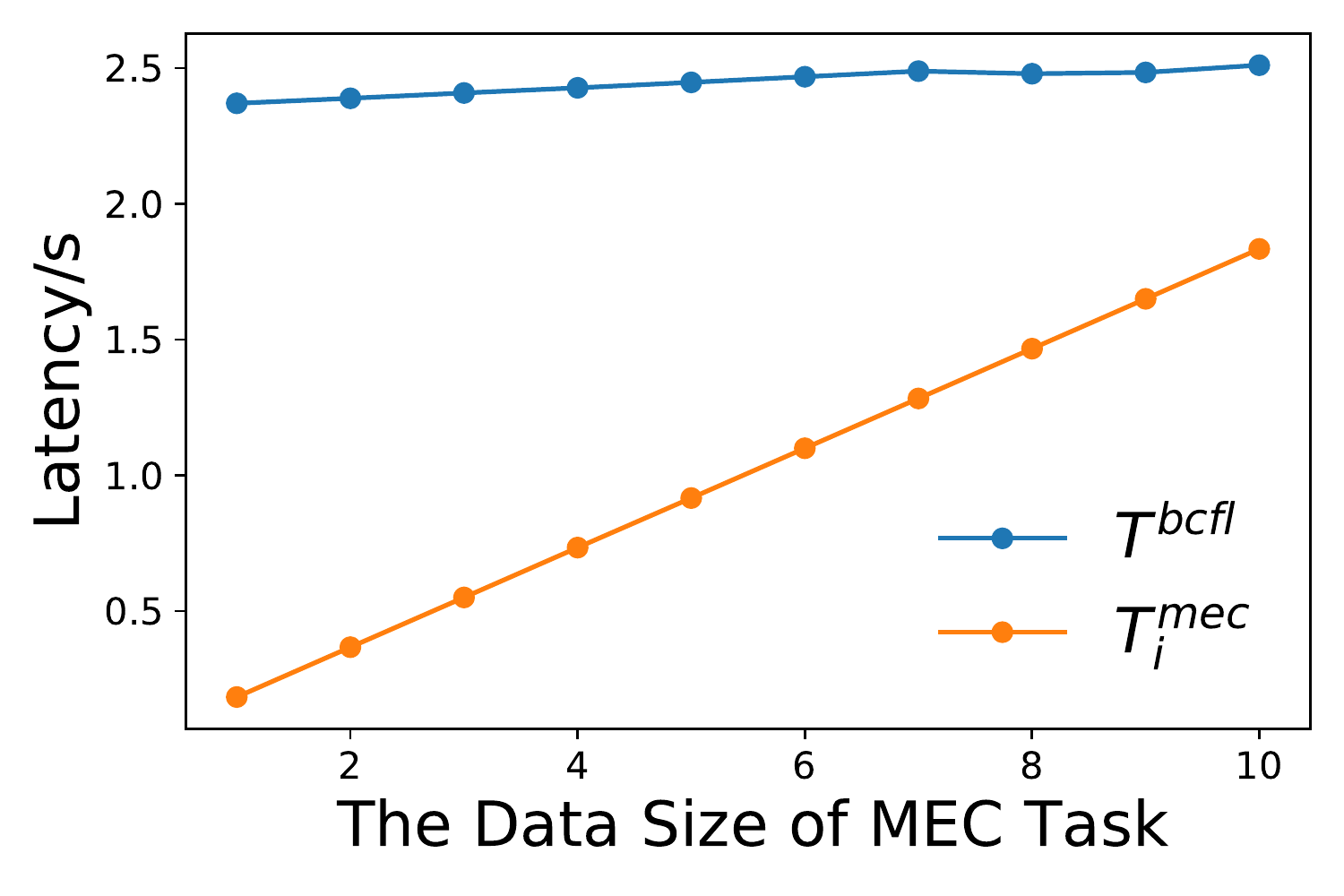}}
\subfigure[Latency changes with the data size of BCFL task.]{
\includegraphics[width=0.23\textwidth]{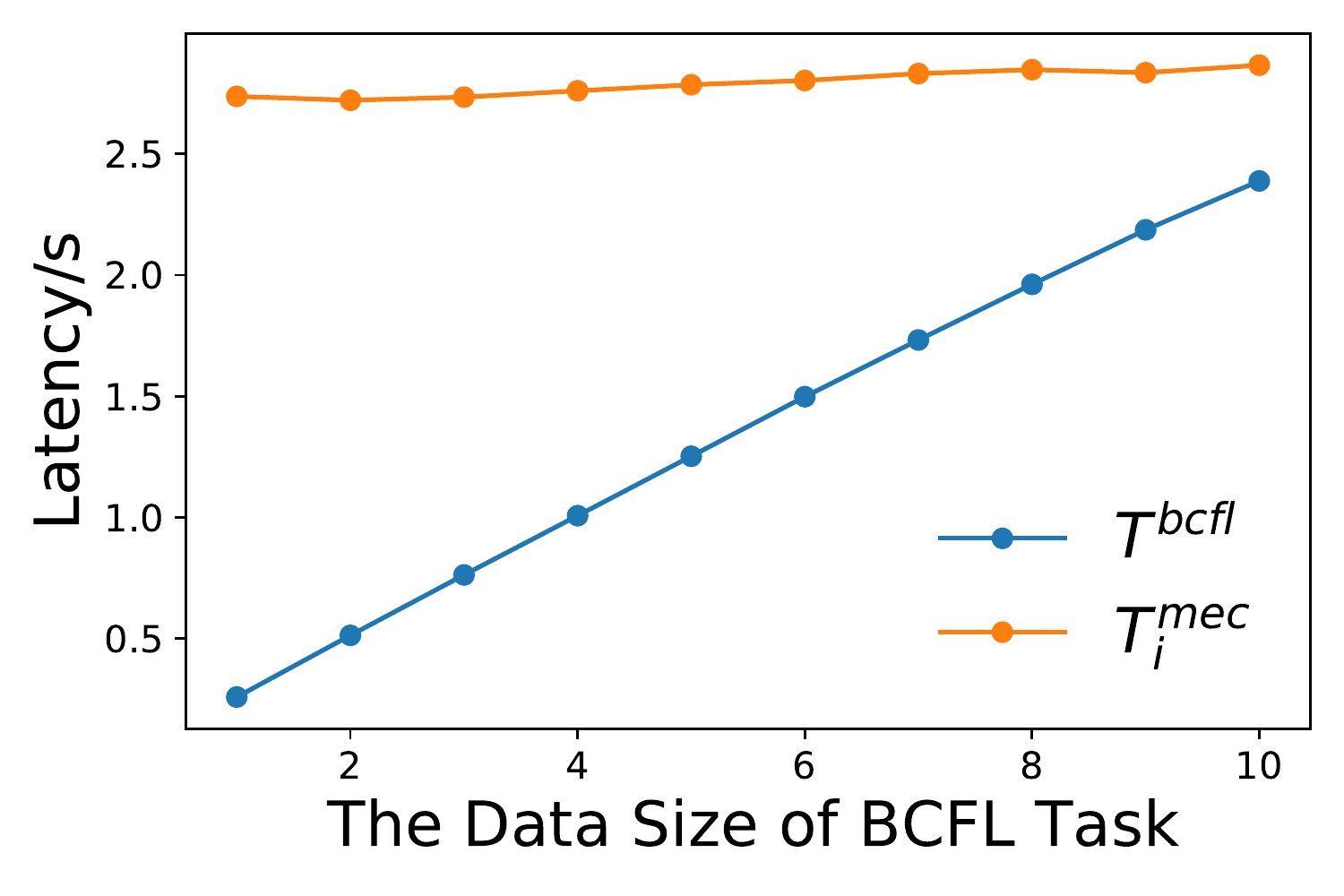}}
\caption{The Latency based on the MG-ADMM algorithm.}
\label{fig_time_1}
\end{figure}

\begin{figure}[h]
\centering
\subfigure[Latency changes with the data size of MEC task.]{
\includegraphics[width=0.23\textwidth]{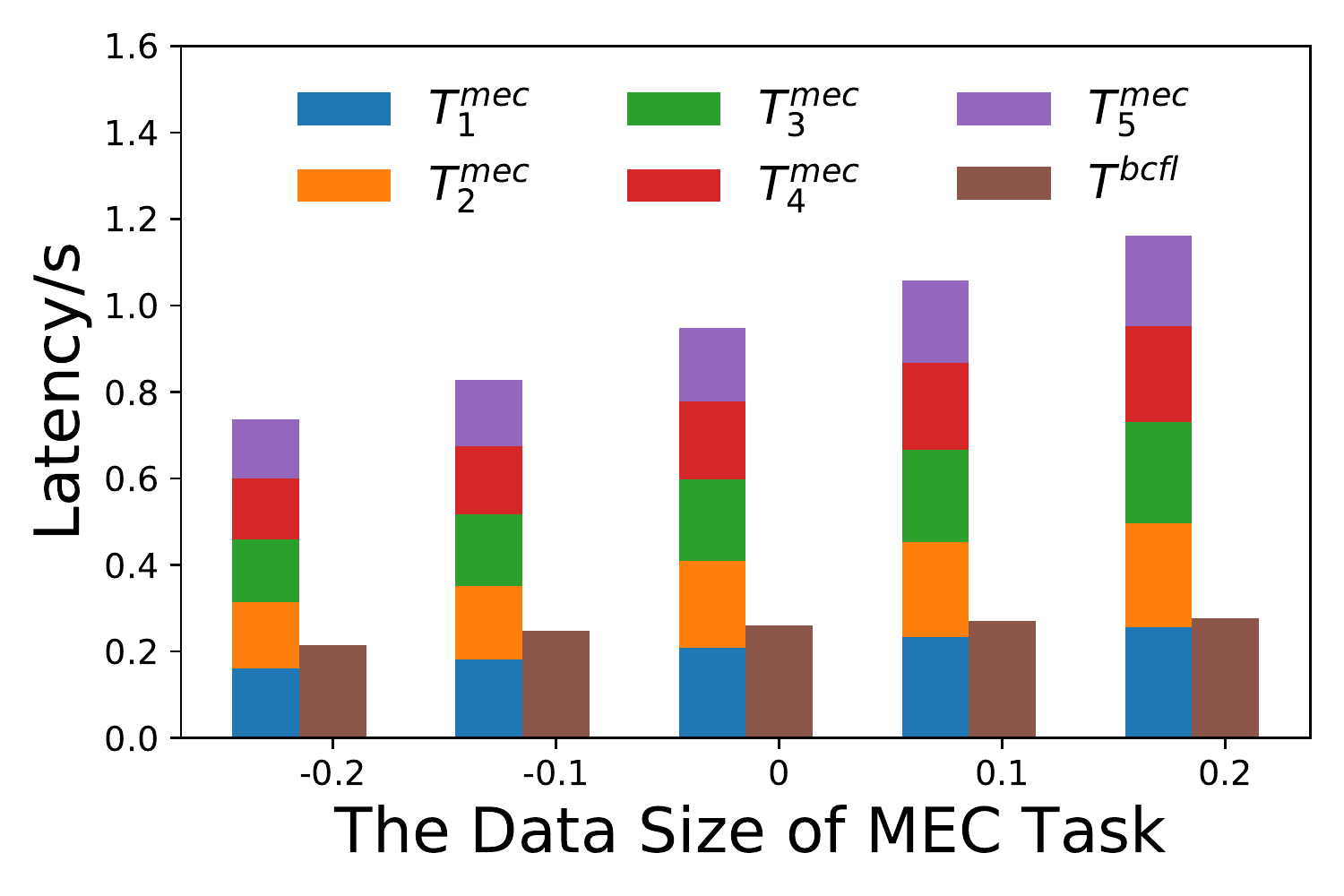}}
\subfigure[Latency changes with the data size of BCFL task.]{
\includegraphics[width=0.23\textwidth]{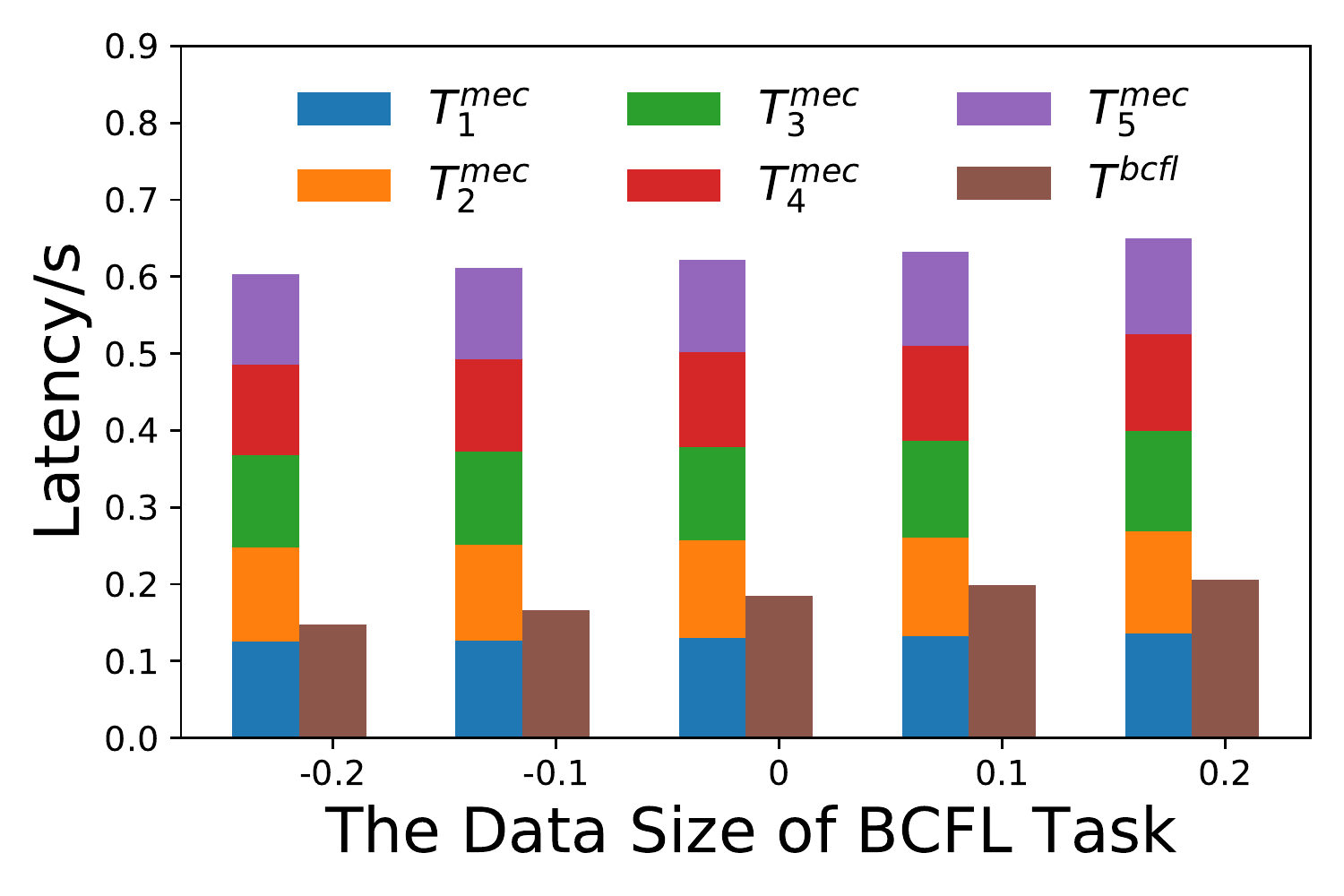}}
\caption{The Latency based on the MC-ADMM algorithm.}
\label{fig_time_2}
\end{figure}

First, we let $T_i^{mec}=T_i^{comm}+T_i^{comp}$ be the total time consumed by the MEC server in processing the MEC task submitted by user $i$ according to the optimal decisions. Similarly, we can define $T^{bcfl}=T_{bcfl}^{comm}+T_{bcfl}^{comp}$ as the time consumption for processing the BCFL task.

Based on the same experimental settings as in Fig. \ref{fig_2}, we calculate the latency of completing both MEC and BCFL tasks. The results based on MG-ADMM are shown in Fig. \ref{fig_time_1}. In Fig. \ref{fig_time_1}(a), we can see that $T^{bcfl}$ increases slightly and $T_i^{mec}$ increases significantly when $D_i$ increases. This is because when the data size of MEC task is larger, more time will be required to complete this task. While less resources will be allocated to process the BCFL task, $T^{bcfl}$ will be also larger. Similarly, we can see the results with the change of $D_{bcfl}$ in Fig. \ref{fig_time_1}(b). 

Then, we analyze the latency of the MC-ADMM algorithm with the same settings as in Fig. \ref{fig_4}. The results are shown in Fig. \ref{fig_time_2}. It is clear that when the data sizes of the MEC and BCFL tasks required to be processed increase, the time spent by the server to complete the tasks also increases.

\section{Related Work}\label{rel}

Recently, there are many studies focusing on deploying BCFL on the edge servers. In \cite{zhao2020mobile}, the authors design a BCFL system running on at the edge with edge servers being responsible for collecting and training the local models, where a device selection mechanism and incentive scheme are proposed to facilitate the performance of the crowdsensing. Rehman \textit{et al.} \cite{ur2020towards} devise a blockchain-based reputation-aware fine-gained FL system to enhance the trustworthiness of devices in the MEC system. The work in \cite{shen2020exploiting} tries to address the privacy protection issue for BCFL in MEC via resisting a novel property inference attack, which attempts to cause the unintended property leakage. Hu \textit{et al.} \cite{hu2021blockchain} deploy a BCFL framework on the MEC edge servers to facilitate finishing mobile crowdsensing tasks, which aims to achieve privacy preservation and incentive rationality at the same time. Qu \textit{et al.} \cite{qu2021chainfl} provide a simulation platform for BCFL in the MEC environment to measure the quality of local updates and configurations of IoT devices. From these studies, it can be concluded that the development of BCFL in MEC is promising, even though there are still some challenges that should be tackled.

Specifically, resource allocation is one of the crucial but open challenges.
Since the resources of edge servers are usually limited, it is essential to design a resource allocation scheme for edge servers to provide satisfactory services for both the MEC and the BCFL tasks with the minimum cost. Wang \textit{et al.} \cite{wang2022incentive} design a joint resource allocation mechanism in BCFL, which assists the participants in deciding the proper resources for completing training and mining tasks. 
In \cite{fan2020hybrid}, a hybrid blockchain-assisted resource trading system is designed to achieve the decentralization and efficiency for FL in MEC. Li \textit{et al.} \cite{li2021blockchain} propose a BCFL framework to tackle the security and privacy challenges of FL, where a computing resource allocation mechanism for training and mining is also designed by optimizing the upper bound of the global loss function. One main vulnerability of this scheme is that all participants are assumed to be homogeneous, which is clearly impractical in the mobile scenario.

In summary, none of the existing studies related to the implementation of BCFL in MEC has ever addressed the resource allocation challenge between the MEC tasks and the BCFL task. Because of the dual roles of edge servers in BCFL and MEC, they have to simultaneously finish the upper-layer BCFL task and provide MEC services for the lower-layer mobile devices. To fill this gap, we devise resource allocation schemes for edge servers in the deployment of BCFL at edge to guarantee the service quality to both sides with the minimum cost.

\section{Conclusion}\label{conc}
In this paper, we are the first to address the resource allocation challenge for edge servers when they are required to handle both the BCFL and MEC tasks. We formulate the design of the resource allocation scheme into a convex, multivariate optimization problem with multiple inequity constraints, and then we design two algorithms based on ADMM to solve it in both the homogeneous and heterogeneous scenarios. Solid theoretic analysis is conducted to prove the validity of our proposed solutions, and numerous experiments are carried out to evaluate the correctness and effectiveness of the algorithms.

%





\ifCLASSOPTIONcaptionsoff
  \newpage
\fi



\bibliographystyle{IEEEtran}
\bibliography{main.bib}
%


%




\begin{IEEEbiography}[{\includegraphics[width=1in,height=1.25in,clip,keepaspectratio]{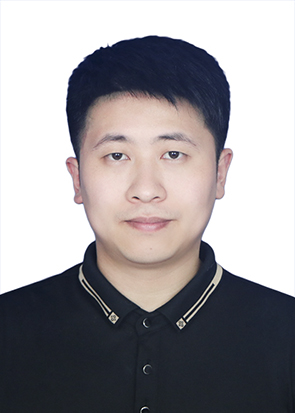}}]{Zhilin Wang} received his B.S. from Nanchang University in 2020. He is currently pursuing his Ph.D. degree of Computer and Information Science In Indiana University-Purdue University Indianapolis (IUPUI). He is a Research Assistant with IUPUI, and he is also a reviewer of 2022 IEEE International Conference on Communications (ICC) and IEEE Access. His research interests include blockchain, federated learning, edge computing, and Internet of Things (IoT).
\end{IEEEbiography}

\begin{IEEEbiography}[{\includegraphics[width=1in,height=1.25in,clip,keepaspectratio]{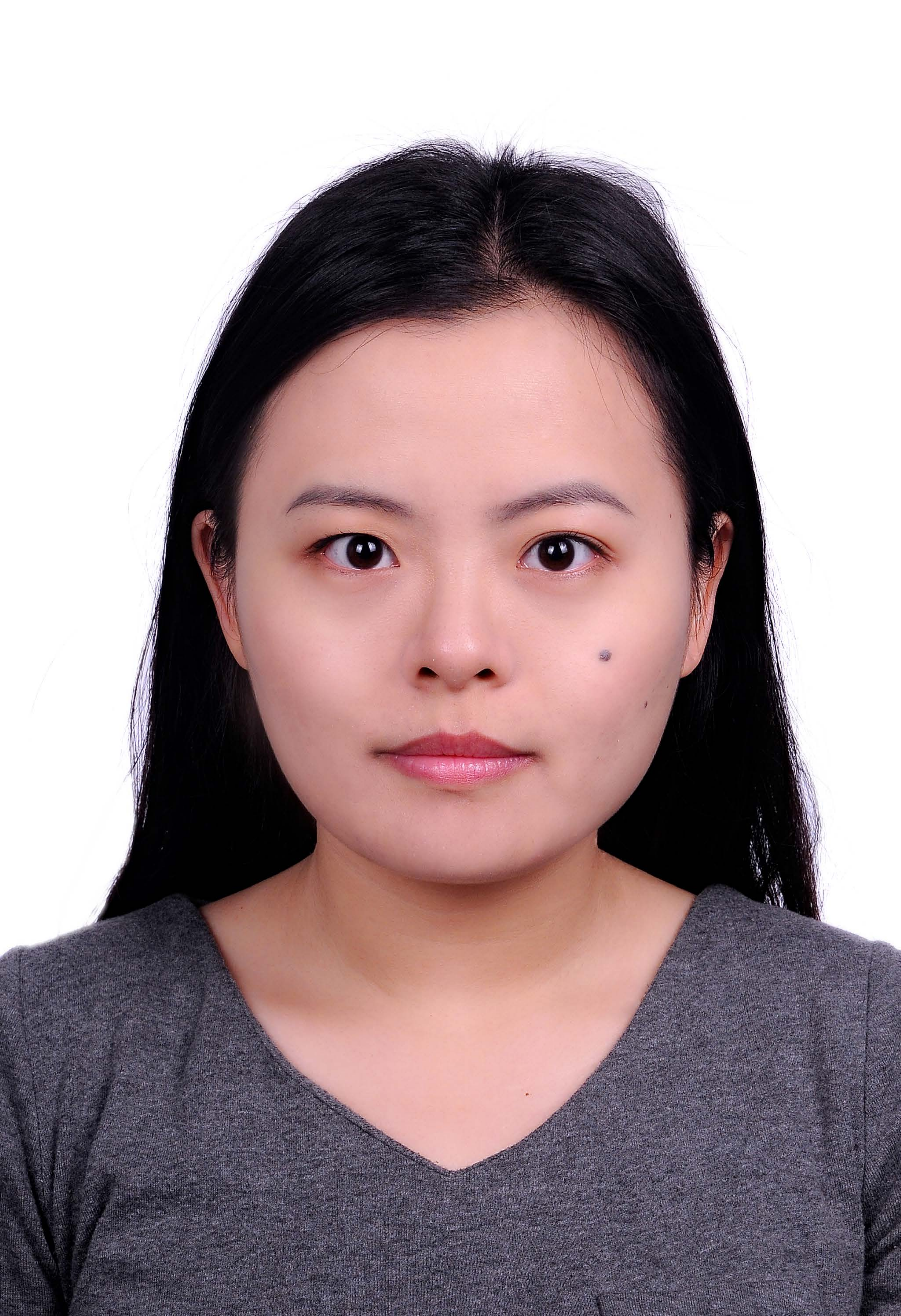}}]{Qin Hu} received her Ph.D. degree in Computer Science from the George Washington University in 2019. She is currently an Assistant Professor with the Department of Computer and Information Science, Indiana University-Purdue University Indianapolis (IUPUI). She has served on the Editorial Board of two journals, the Guest Editor for two journals, the TPC/Publicity Co-chair for several workshops, and the TPC Member for several international conferences. Her research interests include wireless and mobile security, edge computing, blockchain, and crowdsensing.
\end{IEEEbiography}

\begin{IEEEbiography}[{\includegraphics[width=1in,height=1.25in,clip,keepaspectratio]{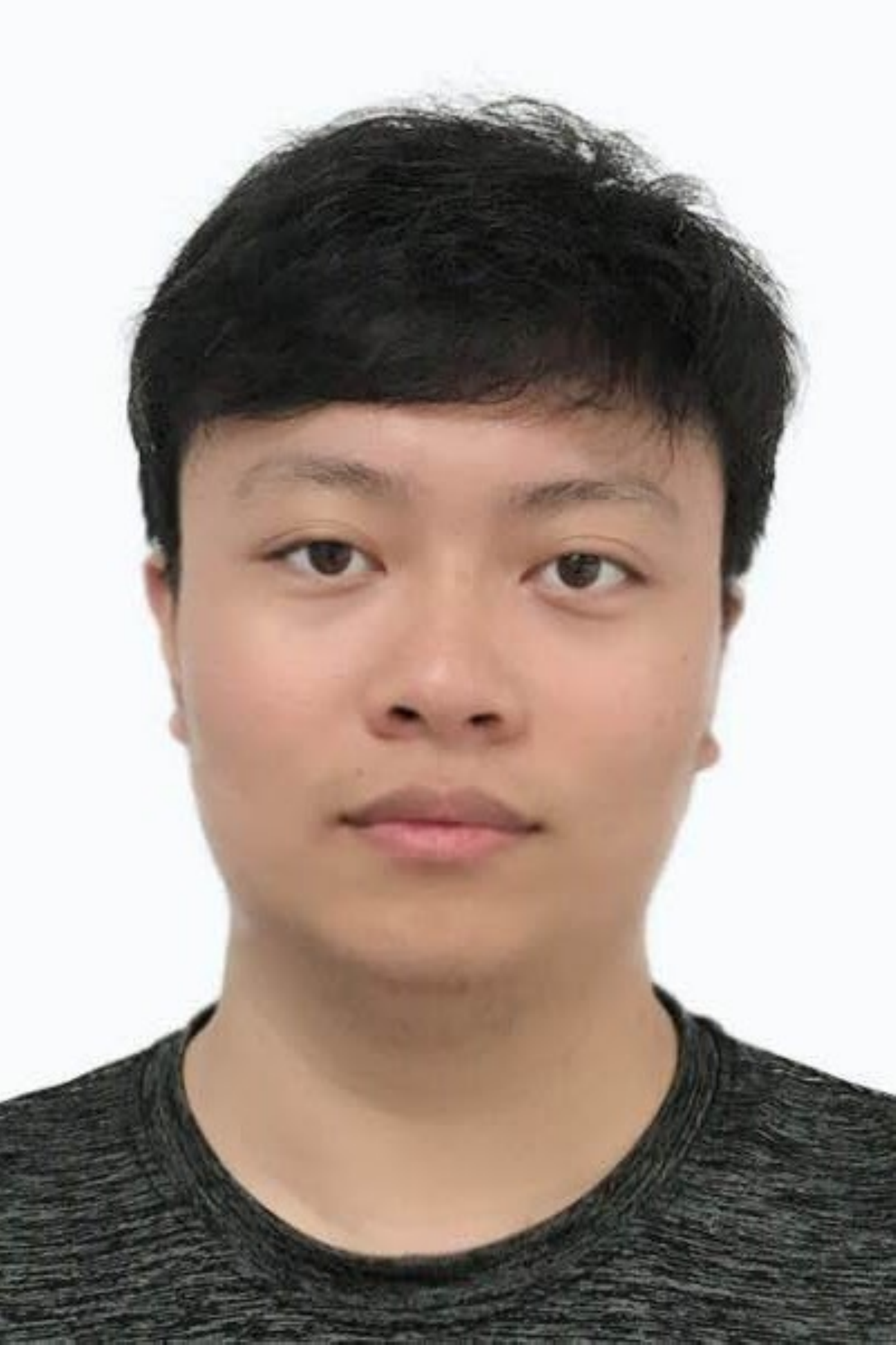}}]{Zehui Xiong} is currently an Assistant Professor in the Pillar of Information Systems Technology and Design, Singapore University of Technology and Design. 
He received the PhD degree in Nanyang Technological University, Singapore. 
His research interests include wireless communications, network games and economics, blockchain, and edge intelligence. He has published more than 140 research papers in leading journals and flagship conferences and many of them are ESI Highly Cited Papers. He has won over 10 Best Paper Awards in international conferences and is listed in the World’s Top $2\%$ Scientists identified by Stanford University. He is now serving as the editor or guest editor for many leading journals including IEEE JSAC, TVT, IoTJ, TCCN, TNSE, ISJ, JAS. 
\end{IEEEbiography}

\end{document}